\documentclass[11pt]{article}
\usepackage[utf8]{inputenc}
\usepackage{amsmath}
\usepackage{amssymb}
\usepackage{amsthm}
\usepackage{tikz}
\usepackage{xcolor}
\usepackage{fullpage}
\usepackage{url}
\usepackage{xfrac}
\usepackage{appendix}
\usepackage{algorithm}
\usepackage{hyperref}

\usepackage{algpseudocode}
\usepackage{natbib}
\hypersetup{
	colorlinks=true,       
	linkcolor=blue,          
	citecolor=blue,        
	filecolor=blue,      
	urlcolor=blue           
}

\numberwithin{equation}{section}
\numberwithin{figure}{section}

\allowdisplaybreaks
\theoremstyle{plain}
\newtheorem{lemma}{Lemma}[section]
\newtheorem{proposition}[lemma]{Proposition}

\newtheorem{theorem}[lemma]{Theorem}

\theoremstyle{definition}
\newtheorem{definition}[lemma]{Definition}
\newtheorem{remark}[lemma]{Remark}

\newtheorem{example}[lemma]{Example}

\definecolor{pythonblue}{RGB}{31 119 180}
\definecolor{pythongreen}{RGB}{44 160 44}
\definecolor{pythonorange}{RGB}{255 127 14}

\newcommand{\dotlab}[2]{\fill (#1) circle (1.1pt); \node[right,inner sep=1.5pt,font=\tiny,text=black] at (#1) {$#2$}}

\newcommand{\treeAlpha}[3]{\begin{scope}[shift={#1},scale=#2,line width=1pt,color=#3]
		\draw (0,0.918)|-(0.35,1.46); \draw (0,0.918)|-(0.7,0.375);
		\draw (0.35,1.46)|-(0.9,1.725); \draw (0.35,1.46)|-(1.4,1.2);
		\draw (0.9,1.725)|-(1.4,1.9); \draw (0.9,1.725)|-(1.4,1.55);
		\draw (0.7,0.375)|-(1.4,0.55); \draw (0.7,0.375)|-(1.4,0.2);
		\dotlab{1.4,1.9}{a};\dotlab{1.4,1.55}{b};\dotlab{1.4,1.2}{c};\dotlab{1.4,0.55}{d};\dotlab{1.4,0.2}{e};
\end{scope}}

\newcommand{\treeBeta}[3]{\begin{scope}[shift={#1},scale=#2,line width=1pt,color=#3]
		\draw (0,0.918)|-(0.35,1.46); \draw (0,0.918)|-(0.7,0.375);
		\draw (0.35,1.46)|-(0.9,1.725); \draw (0.35,1.46)|-(1.4,1.2);
		\draw (0.9,1.725)|-(1.4,1.9); \draw (0.9,1.725)|-(1.4,1.55);
		\draw (0.7,0.375)|-(1.4,0.55); \draw (0.7,0.375)|-(1.4,0.2);
		\dotlab{1.4,1.9}{a};\dotlab{1.4,1.55}{c};\dotlab{1.4,1.2}{b};\dotlab{1.4,0.55}{d};\dotlab{1.4,0.2}{e};
\end{scope}}

\newcommand{\treeGamma}[3]{\begin{scope}[shift={#1},scale=#2,line width=1pt,color=#3]
		\draw (0,1.10)|-(0.35,1.64); \draw (0,1.10)|-(0.7,0.375);
		\draw (0.35,1.64)|-(1.4,1.9); \draw (0.35,1.64)|-(0.9,1.375);
		\draw (0.9,1.375)|-(1.4,1.55); \draw (0.9,1.375)|-(1.4,1.2);
		\draw (0.7,0.375)|-(1.4,0.55); \draw (0.7,0.375)|-(1.4,0.2);
		\dotlab{1.4,1.9}{a};\dotlab{1.4,1.55}{b};\dotlab{1.4,1.2}{c};\dotlab{1.4,0.55}{d};\dotlab{1.4,0.2}{e};
\end{scope}}

\newcommand{\treeSigma}[3]{\begin{scope}[shift={#1},scale=#2,line width=1pt,color=#3]
		\draw (0,0.9)|-(0.35,1.55);            
		\draw (0,0.9)|-(0.7,0.375);            
		\draw (0.35,1.55)|-(1.4,1.9);          
		\draw (0.35,1.55)|-(1.4,1.55);
		\draw (0.35,1.55)|-(1.4,1.2);
		\draw (0.7,0.375)|-(1.4,0.55);         
		\draw (0.7,0.375)|-(1.4,0.2);
		\dotlab{1.4,1.9}{a};\dotlab{1.4,1.55}{b};\dotlab{1.4,1.2}{c};\dotlab{1.4,0.55}{d};\dotlab{1.4,0.2}{e};
		\node[text=black] at (0.16,1.78) {\tiny$\ast$};   
\end{scope}}

\newcommand{\Rbb}{\mathbb{R}}
\newcommand{\Zbb}{\mathbb{Z}}
\newcommand{\Nbb}{\mathbb{N}}
\newcommand{\diff}{\,\textnormal{d}}
\DeclareSymbolFont{bbold}{U}{bbold}{m}{n}
\DeclareSymbolFontAlphabet{\mathbbold}{bbold}
\newcommand{\ind}{\mathbbold{1}}

\newcommand{\Pbb}{\mathbb{P}}
\newcommand{\Ebb}{\mathbb{E}}
\newcommand{\F}{\mathcal{F}}

\renewcommand{\S}{\mathcal{S}}
\newcommand{\B}{\mathcal{B}}
\renewcommand{\P}{\mathcal{P}}
\newcommand{\TT}{\mathcal{T}}

\begin{document}
		
	\title{Phylogenetic Inference and the Stickiness of Fréchet Means, via Precise Asymptotics of an Embedded Random Walk}
	
	\author{Adam Quinn Jaffe\thanks{Department of Statistics, Columbia University, New York, NY, \texttt{a.q.jaffe@columbia.edu}}}

	\date{}
			
	\maketitle
	
	\begin{abstract}
		A well-known phenomenon in statistical analyses of populations of phylogenetic trees in the Billera-Holmes-Vogtmann space is that the topology of the Fr\'echet mean tree can contain multifurcations (i.e., internal nodes with more than two children), which raises the practical question of whether this reflects a population-level branching structure (hard polytomy) or merely sampling variability in the data (soft polytomy).
		This is an instance of the more general phenomenon of ``stickiness'' in non-Euclidean statistics, whereby the sample Fr\'echet mean in certain non-positively curved stratified spaces becomes permanently trapped in a lower-dimensional stratum.
		In this work, we identify a particular multidimensional random walk embedded within the Fr\'echet mean process, and we show that the time at which stickiness occurs is determined by the largest last-passage time above zero of the coordinates of this random walk.		
		Using this representation, we develop a fully nonparametric procedure for estimating the probability that trifurcations in a sample Fr\'echet mean tree will bifurcate at some future time if more observations are collected.
		Lastly, we apply our methodology to a problem in phylogenetics where we consider whether an observed trifurcation in the species tree of primates, glires, and tree shrews is genuinely trifurcated at the population level.
	\end{abstract}
	
	\medskip
	
	\small
	
		
	\medskip
	

	\normalsize

	\section{Introduction}
	
	\begin{figure}[t]
		\centering
		\includegraphics[scale=0.75]{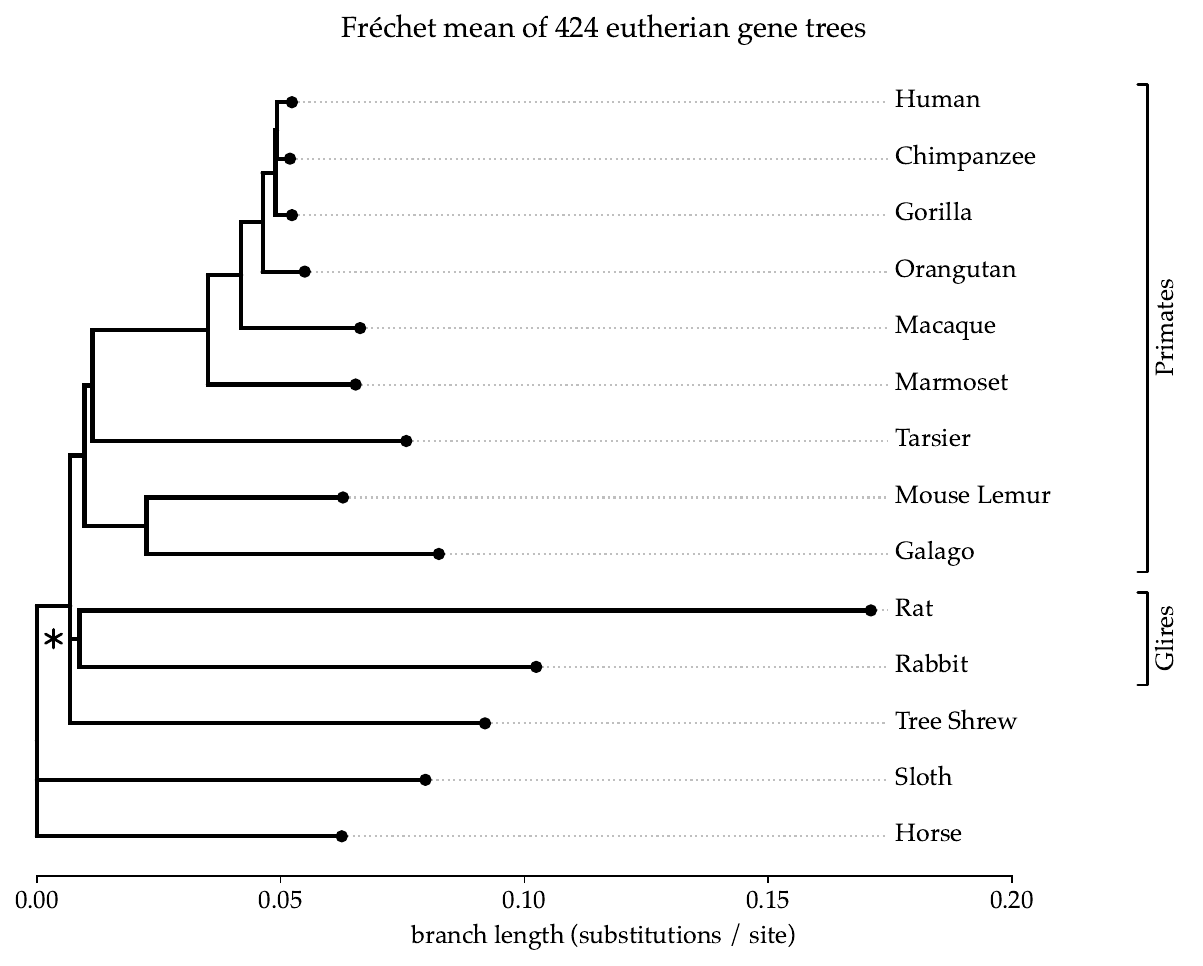}
		\caption{The Fr\'echet mean of many gene trees for a group of eutherian mammals.
		We construct gene trees for $14$ different mammals across $424$ different regions of their genomes, and we show their Fr\'echet mean tree in the Billera-Holmes-Vogtmann treespace.
		The Fr\'echet mean tree has a trifurcation among $\{\textnormal{Tree Shrew, Glires, Primates}\}$ denoted by $\ast$.}
		\label{fig:primate_mean}
	\end{figure}

In the field of non-Euclidean statistics \cite{Huckemann2015,dryden2009non}, an important object of study is the Fr\'echet mean, which provides a canonical notion of central tendency for non-Euclidean data by generalizing the notions of sample mean and expectation to metric spaces which may lack a vector space structure.
For Fr\'echet means of independent identically-distributed (i.i.d.) samples, much is known about their convergence, in a probabilistic sense, to a population counterpart, including laws of large numbers \cite{Sverdrup, Ziezold, Schoetz2, EvansJaffeSLLN, JaffeInfDim}, central limit theorems \cite{ManifoldsI, ManifoldsII, FrechetCLTs}, large deviations principles \cite{JaffeSantoroLDP}, and various concentration inequalities \cite{AGP, Schoetz1}.

One of the motivating examples for this field is the Billera-Holmes-Vogtmann (BHV) treespace \cite{BHV}, which is a non-positively curved geometry on the set of phylogenetic trees with a fixed set of leaves.
(See Subsection~\ref{subsec:BHV-def} for precise definitions.)
This example is of interest in biology since the Fr\'echet mean in BHV treespace provides a natural method for aggregating a collection of gene trees into a single estimated species tree, which was one of the original goals of the seminal work \cite{BHV}.
However, Fr\'echet means in BHV treespace are known to exhibit some peculiar phenomena; for instance, Figure~\ref{fig:primate_mean} shows the Fr\'echet mean tree of $424$ gene trees for $14$ eutherian mammals, which exhibits a trifurcation among \{Tree Shrew, Glires, Primates\}.
It is natural to ask whether an empirical trifurcation reflects a genuinely unresolved relationship at the population level, or merely sampling variability (referred to as \textit{hard} and \textit{soft polytomy}, respectively, in biology \cite{SayyariMirarab}).
This question is especially salient for the trifurcation above, given the ongoing debate in evolutionary biology about whether tree shrews are more closely related to glires or to primates (e.g., \cite{SongMammals,TreeShrewsPosition,LinShrews}).
Our goal in this paper is to develop statistical tools for inference related to trifurcations like those above.

Trifurcations in BHV treespace are an instance of a more general phenomenon in non-Euclidean statistics known as \textit{stickiness}, whereby the sample Fr\'echet mean in certain spaces of negative curvature becomes eventually trapped in a subset of low dimension.
This contrasts with the behavior in Euclidean spaces, where the sample mean converges to the population mean with non-trivial (usually Gaussian) fluctuations.
Stickiness has been studied in many works (e.g., \cite{Kale,UlmerVanHuckemann,StickyOpenBooks}), usually in the setting of \textit{stratified spaces} which, roughly speaking, are defined by gluing together Riemannian manifolds along submanifolds of smaller dimension; besides BHV treespaces, other prominent examples of stratified spaces in non-Euclidean statistics include spaces of covariance matrices \cite{Takatsu,dryden2009non} and Kendall's shape spaces \cite{Huckemann,le2000frechet}.
As explained in \cite[p. 2240]{StickyOpenBooks}, when the population Fr\'echet mean in a stratified space is separated from all singularities, its asymptotic theory is essentially equivalent to that of Fr\'echet means on Riemannian manifolds; if instead the population Fr\'echet mean lies on a singularity, then qualitatively new phenomena like stickiness can occur.

In the mathematical developments of this work, we consider the setting of \textit{open books}, roughly defined as follows.
(See Subsection~\ref{subsec:def} for precise definitions.)
Fix $K\ge 3$ and $m\ge 1$ and let $\B_{K,m}$ denote the metric space resulting from gluing together $K$-many copies of $[0,\infty)\times\Rbb^{m-1}$ along their common boundary $\S$.
Then suppose that $Y_1,\ldots, Y_n$ are i.i.d. samples from some probability measure $\mu$ on $\B_{K,m}$, and let $F_n$ denote the Fr\'echet mean of $Y_1,\ldots, Y_n$ for all $n\in\Nbb$.
Stickiness is precisely the statement that the random variable $T:=\inf\{n\in\Nbb: F_m\in\S \textnormal{ for all }m\ge n\}$ is finite almost surely, and our goal is to understand finer details of the random variable $T$.
(Note that the random time $T$ is not a stopping time with respect to the filtration generated by $Y_1,Y_2,\ldots$, which partially explains why analyzing $T$ is difficult.)
Open books capture the local geometry of BHV treespaces and other stratified spaces near their codimension-one strata, so studying stickiness in the setting of open books is a standard reduction in the literature.

The first main result of the paper identifies the precise asymptotics of $\Pbb(T>n)$ in a general setting.
That is, under some mild conditions, we have (Theorem~\ref{thm:sticking-time-tail}), as $n\to\infty$,
\begin{equation*}
	\frac{c}{\sqrt{n}}(1+o(1))e^{-\alpha n}\le \Pbb(T>n) \le \frac{C}{\sqrt{n}}(1+o(1))e^{-\alpha n},
\end{equation*}
where $c,C,\alpha>0$ are constants that depend only on $\mu$, and which we explicitly calculate in several concrete examples.
This result is sharp in the sense that we identify the dependence of $\Pbb(T>n)$ on $n$ to be precisely of order $n^{-1/2}e^{-\alpha n}$; moreover, the constants $c$ and $C$ are explicit, and when the rate $\alpha$ is attained on a single page they differ only by a factor of the form $\theta/(\theta-\gamma)$ for two explicit parameters $0<\gamma<\theta$ determined by $\mu$.

The second main result is a method for estimating $\Pbb_{\mu}(T>n\,|\,Y_1,\ldots,Y_n)$ from $Y_1,\ldots, Y_n$ when $\mu$ is unknown; in other words, we aim to estimate the probability that stickiness has not yet occurred, given the data observed so far.
More precisely, we define a fully nonparametric procedure (Algorithm~\ref{alg:unstick}) which we prove (Proposition~\ref{prop:alg-consistency}) achieves a vanishing relative error of the order $O_{\Pbb}(n^{-1/2})$.
In practice this estimator is downwardly-biased, so we also define a bootstrap bias-corrected version of the procedure which appears to have slightly better performance.

	\begin{figure}[h!]
		\includegraphics[scale=0.75]{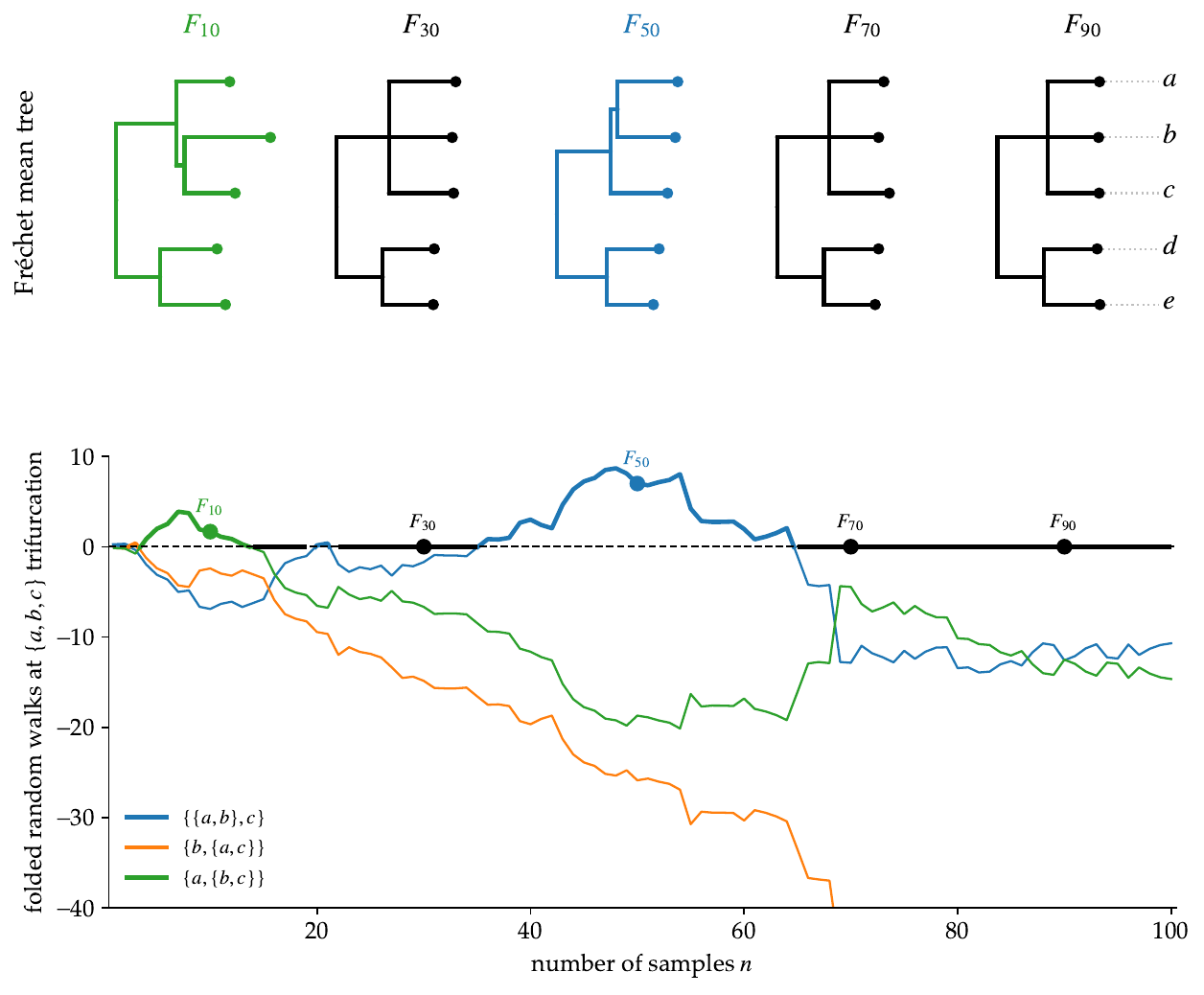}
		\caption{Random walks embedded in the Fr\'echet mean of samples in the Billera-Holmes-Vogtmann (BHV) treespace, where stickiness is determined by the largest last passage time above zero.		In a simulation on $N=5$ leaves, we show (top) the Fr\'echet mean tree $F_n$ for $n\in\{10,30,50,70,90\}$, and (bottom) the first coordinates of the folded random walks.}
		\label{fig:treespace-sim}
	\end{figure}

While these mathematical results are limited to the case of open books, we describe (Section~\ref{sec:BHV}) the minor modifications needed in order to apply them to the case of BHV treespaces near their top-dimension strata; in this setting, the codimension-one strata are exactly the sets of trees with a trifurcation.
Consequently, we revisit the example from Figure~\ref{fig:primate_mean}, and we use the bootstrap bias-corrected procedure above to estimate the probability that the observed trifurcation among \{Tree Shrew, Glires, Primates\} is a soft polytomy (i.e., that it will eventually resolve to binary branching if more data are collected).

All of our mathematical results are based on one fundamental observation, which we view as the central insight of this work.
To state it, recall that the open book admits \textit{folding maps} $\Phi_j:\B_{K,m}\to \Rbb^m$ for each $1\le j\le K$, visualized in Figure~\ref{fig:spider-book}.
It is known (see \cite[Lemma~3.3]{StickyOpenBooks} and \cite[Lemma~4]{UlmerVanHuckemann}) that the Fr\'echet mean satisfies $F_n\in\S$ if and only if we have $\sum_{i=1}^{n}(\Phi_j(Y_i))_1\le 0$ for all $1\le j\le K$, where $v_1$ denotes the first coordinate of $v\in\Rbb^m$; this identity has been previously exploited for each fixed $n\in\Nbb$, but we will instead consider it in a pathwise manner.
In particular, we observe that the stochastic process $\{S_{j,n}\}_{n\in\Nbb}$ defined via $S_{j,n}:=\sum_{i=1}^{n}(\Phi_j(Y_i))_1$ is a random walk, and stickiness is precisely the statement that the random walks $S_1,\ldots, S_K$ are all eventually non-positive almost surely.
Then, our first result follows by controlling simultaneously the last passage times of these random walks, and our second result follows by estimating the large deviations exponent of each constituent random walk.

This insight also holds in the BHV treespace, where the folding maps take a more complicated form (see Figure~\ref{fig:BHV-folding}) and there is a corresponding identity for trifurcations in the Fr\'echet mean tree (see \cite[Theorem~3]{BardenLe} and \cite[Section~4]{BardenOwenLe}); we illustrate the insight in Figure~\ref{fig:treespace-sim}, which depicts the Fr\'echet mean tree process alongside the induced random walks.
We believe that this insight will also be useful for developing other statistical methodologies in open books and in BHV treespaces, for instance determining the sample complexity of estimating the topology of the Fr\'echet mean tree.
	All figures in the paper can be replicated with the publicly available code at \url{https://github.com/aqjaffe/PhyloStickyRW}.
	
	\section{Preliminaries on Random Walks}\label{sec:RW}
	
	In this section we introduce some background that will be used in our main results.
	Throughout, we assume that $\nu$ is some probability measure on $\Rbb$ and that $X_1,X_2,\ldots$ are i.i.d. samples from $\nu$, on some probability space with probability measure denoted $\Pbb_{\nu}$ and expectation denoted $\Ebb_{\nu}$, and we write $X$ for an additional copy of this random variable.
	We define the \textit{random walk (RW) with step distribution $\nu$} as the stochastic process $S=\{S_n\}_{n\in\Nbb}$ defined via $S_n:=\sum_{i=1}^{n}X_i$, which satisfies $S_0=0$ by convention.
	We assume $\Ebb_{\nu}|X|<\infty$ and $\Ebb_{\nu}[X] < 0$, so that the law of large numbers implies $S_n/n\to \Ebb_{\nu}[X]<0$; this yields $S_n\to -\infty$ almost surely, meaning the walk has negative drift and is transient.
	In the remainder of this section we give some results which quantitatively describe this divergence.
	
	All of these results involve the \textit{cumulant generating function (CGF)} of $\nu$, which is the function $\Lambda_\nu:\Rbb\to\Rbb\cup\{\infty\}$ defined via
	\begin{equation*}
		\Lambda_\nu(\lambda):=\log \mathbb E_\nu[e^{\lambda X}].
	\end{equation*}	
	As explained in \cite[Chapter~2]{LargeDeviations}, the function $\Lambda_\nu$ is convex and lower semicontinuous, although it may take the value $\infty$ for some parameters $\lambda\in\Rbb$.
	Throughout, we additionally assume the \textit{Cram\'er condition} that there exists $\theta_{\nu}>0$ with $\Lambda_{\nu}(\theta_{\nu})=0$.
	By convexity and $\Lambda_\nu(0)=0$, such $\theta_{\nu}$ is necessarily unique.
	Also, $\Lambda_\nu$ is finite on $[0,\theta_\nu]$, since $e^{\lambda x}\le 1+e^{\theta_\nu x}$ for all $\lambda\in[0,\theta_\nu]$ and $x\in\Rbb$, hence $\Lambda_\nu$ is continuous on $[0,\theta_\nu]$ and infinitely differentiable on $(0,\theta_\nu)$ \cite[Lemma~2.2.5 and Exercise~2.2.24]{LargeDeviations}.
	Moreover, since $(e^{\lambda X}-1)/\lambda$ decreases to $X$ as $\lambda\downarrow 0$ and is integrable at $\lambda=\theta_\nu$, the monotone convergence theorem shows that the right derivative of $\Lambda_\nu$ at zero equals $\Ebb_{\nu}[X]<0$.
		
	The first result we will need is a description of the probability that $S$ deviates from its typical value at a fixed time $n$, which is the domain of large deviations theory (see \cite{LargeDeviations}).
	More precisely, we have $S_n\to -\infty$ almost surely, hence $\Pbb_{\nu}(S_n>0)\to 0$ as $n\to\infty$, but it is of interest to understand the precise speed at which this probability decays as a function of $n$.
	One answer to this question is \textit{Cram\'er's theorem} \cite[Theorem~2.2.3]{LargeDeviations} which states that the decay is always exponentially fast, and that the rate of exponential decay can be described as follows.
	Since $\Lambda_\nu'(0)=\Ebb_\nu[X]<0$ and $\Lambda_\nu$ is strictly convex and coercive, it has a unique minimizer which we denote $\gamma_\nu$ and which must lie in $(0,\infty)$; equivalently, it is the unique positive solution to $\Lambda_\nu'(\lambda)=0$.
	Then, we have
	\begin{equation}\label{eqn:Cramer-thm}
		\frac{1}{n}\log\Pbb_{\nu}(S_n>0)\to \inf_{\lambda\in\Rbb}\Lambda_{\nu}(\lambda)= -(-\Lambda_{\nu}(\gamma_{\nu})),
	\end{equation}
	as $n\to\infty$.
	Note that $\Lambda_{\nu}(0) = 0$ and $\Lambda_{\nu}'(0)<0$, so we have $\inf_{\lambda\in\Rbb}\Lambda_{\nu}(\lambda)< 0$ hence $-\Lambda_{\nu}(\gamma_{\nu})>0$.

	In fact, the previous result can be made sharper.
	Instead of asking to determine the limiting value of $n^{-1}\log \Pbb_{\nu}(S_n>0)$, we can ask to determine a sequence $p_n\to 0$, as an explicit function of $\nu$ and $n$, which satisfies $\Pbb_{\nu}(S_n>0)= p_n(1+o(1))$.
	The answer depends on the arithmetic structure of $\nu$; we say that $\nu$ is \emph{lattice} if there exist $a\in\Rbb$ and $d>0$ such that $\nu(a+d\Zbb)=1$, and that $\nu$ is \emph{non-lattice} otherwise.
	If $\nu$ is non-lattice, then the \textit{Bahadur-Rao theorem} \cite[Theorem~3.7.4]{LargeDeviations} states
	\begin{equation}\label{eqn:BH-thm}
		\Pbb_\nu(S_n>0)
		=\frac{c_{\nu}}{\sqrt{n}}(1+o(1))
		e^{-n(-\Lambda_\nu(\gamma_{\nu}))}
		\qquad\textnormal{where}\qquad
		c_{\nu}:=\frac{1}{\gamma_{\nu}\sqrt{2\pi \Lambda_{\nu}''(\gamma_{\nu})}}.
	\end{equation}
	In other words, the decay of $\Pbb_{\nu}(S_n>0)$ has exponential speed $e^{-n(-\Lambda_{\nu}(\gamma_{\nu}))}$ given by Cram\'er's theorem and a polynomial correction factor $c_{\nu}n^{-1/2}$ given by the Bahadur-Rao theorem.
	If $\nu$ is lattice, then \cite[Theorem~3.7.4]{LargeDeviations} still describes $\Pbb_{\nu}(S_n\ge 0)$ along those $n$ for which $0$ lies in the support of $S_n$, albeit with a different constant, but $\Pbb_{\nu}(S_n>0)$ need not admit asymptotics of the form~\eqref{eqn:BH-thm}; see Example~\ref{ex:K-pt} for an illustration.
	
		The second result we will need is a description of the probability that $S$ deviates from its typical behavior on an infinite time horizon; more precisely, we aim to compute the probability that the RW ever returns to a positive value if it is started at a negative value, which is the domain of ruin theory (see \cite{CramerLundberg}).
	To set this up, define the \textit{first passage time} as $\tau:=\inf\{n\ge0:S_n>0\}$, with the convention $\inf\varnothing=\infty$, and define the function $h_{\nu}(x) = \Pbb_{\nu}(\tau<\infty\,|\,S_0=x)$ for $x\in\Rbb$.
	If in addition to the existence of the Cram\'er root $\theta_{\nu}>0$ we also have $\Ebb_{\nu}[Xe^{\theta_{\nu}X}]<\infty$, then a consequence of the classical \textit{Cram\'er-Lundberg theory} \cite[Chapter~I]{CramerLundberg} is
	\begin{equation}\label{eqn:CR-limit}
		\frac{1}{x}\log h_{\nu}(x) \to \theta_{\nu}
	\end{equation}
	as $x\to -\infty$.
	Again, note the sign convention; $\theta_{\nu}>0$ by convention and $x< 0$, so $x\theta_{\nu}< 0$.
	This result can indeed be sharpened by determining a function $q(x)\to 0$, as an explicit function of $\nu$ and $x$, which satisfies $h_{\nu}(x)= q(x)(1+o(1))$, but we will not need the details in this paper.
	Instead, we will also use \textit{Lundberg's bound} which states
	\begin{equation}\label{eqn:Lundberg-bd}
		h_{\nu}(x) \le e^{\theta_{\nu}x}
	\end{equation}
	for all $x\le 0$.
	Also, note that $h_{\nu}(x)=1$ for $x\ge 0$.
	
	We also define the \textit{last passage time} as $\sigma:=\sup\{n\ge0:S_n>0\}$, with the convention that $\sup\varnothing =0$, whose tail events $\{\sigma>n\}$ will play a central role in the sequel for various choices of distributions $\nu$.
	Indeed, we will consider such tail events for a coupled collection of random walks, and we will use the preliminaries in this section to analyze their probabilities in terms of the CGF $\Lambda_{\nu}$ and the distinguished parameters $\gamma_\nu>0$ and $\theta_\nu>0$.

	\section{Main Results in Open Books}
	
	In this section we prove the main results of the paper.
	We begin in Subsection~\ref{subsec:def} where we introduce the basic definitions of open books, Fr\'echet means, and stickiness.
	Then, Subsection~\ref{subsec:represent} contains our representation of the sticking time in terms of suitable random walks, Subsection~\ref{subsec:asymp} develops precise asymptotics for the sticking time, and Subsection~\ref{subsec:estim} develops a procedure (and a bootstrap bias-corrected variation) for estimating the probability of future unsticking.
	
	\subsection{Basic Definitions}\label{subsec:def}
	
		\begin{figure}
		\centering
		\begin{tikzpicture}
			
			\def\shift{5}
			
			\filldraw[thin, opacity=0.65, color=black!30] (-\shift, 0) to (-\shift, 3) to (-1.75-\shift,3) to (-1.75-\shift, 0) to cycle;
			\draw[thin, color=black] (-\shift, 0) to (-\shift, 3) to (-1.75-\shift,3) to (-1.75-\shift, 0) to cycle;
			\filldraw[thin, opacity=0.65, color=black!30] (-\shift, 0) to (-\shift, 3) to (1.5-\shift,4) to (1.5-\shift, 1) to cycle;
			\draw[thin, color=black] (-\shift, 0) to (-\shift, 3) to (1.5-\shift,4) to (1.5-\shift, 1) to cycle;
			\filldraw[color=pythonorange] (0.5-\shift, 2) circle (2pt);
			\filldraw[thin, opacity=0.65, color=black!30] (-\shift, 0) to (-\shift, 3) to (1.25-\shift,2) to (1.25-\shift, -1) to cycle;
			\draw[thin, color=black] (-\shift, 0) to (-\shift, 3) to (1.25-\shift,2) to (1.25-\shift, -1) to cycle;
			
			\node at (-0.5-\shift,3.5) {$\mathcal{B}_{3,2}$};
			\node at (-2.25-\shift,0.5) {$\mathcal{P}_{1}$};
			\node at (1.75-\shift,0) {$\mathcal{P}_{2}$};
			\node at (2-\shift,1.5) {$\mathcal{P}_{3}$};
			
			\filldraw[color=pythonblue] (-1.25-\shift, 1.25) circle (2pt);
			\filldraw[color=pythongreen] (1-\shift, 0.25) circle (2pt);

			\filldraw[thin, opacity=0.65, color=black!30] (0, 0) to (0, 3) to (-1,4) to (-1, 1) to cycle;
			\draw[thin, color=black] (0, 0) to (0, 3) to (-1,4) to (-1, 1) to cycle;
			\filldraw[thin, opacity=0.65, color=black!30] (0, 0) to (0, 3) to (1.5,2.5) to (1.5, -0.5) to cycle;
			\draw[thin, color=black] (0, 0) to (0, 3) to (1.5,2.5) to (1.5, -0.5) to cycle;
			\filldraw[color=pythonorange] (-0.35, 2) circle (2pt);
			\filldraw[thin, opacity=0.65, color=black!30] (0, 0) to (0, 3) to (-1.75,3) to (-1.75, 0) to cycle;
			\draw[thin, color=black] (0, 0) to (0, 3) to (-1.75,3) to (-1.75, 0) to cycle;
			\node[color=white] at (0,-0.75) {$\mathcal{B}_{3,2}$};

			\filldraw[color=pythonblue] (-1.25, 1.25) circle (2pt);
			\filldraw[color=pythongreen] (1.25, 0.5) circle (2pt);

			\filldraw[thin, opacity=0.65, color=black!30] (\shift, 0) to (\shift, 3) to (-1.75+\shift,3) to (-1.75+\shift, 0) to cycle;
			\draw[thin, color=black] (\shift, 0) to (\shift, 3) to (-1.75+\shift,3) to (-1.75+\shift, 0) to cycle;
			\filldraw[thin, opacity=0.65, color=black!30] (\shift, 0) to (\shift, 3) to (1.75+\shift,3) to (1.75+\shift, 0) to cycle;
			\draw[thin, color=black] (\shift, 0) to (\shift, 3) to (1.75+\shift,3) to (1.75+\shift, 0) to cycle;
			\node[color=white] at (\shift,-0.75) {$\mathcal{B}_{3,2}$};
			\node at (-1.5+\shift,3.35) {\footnotesize$(-\infty,0]\times\Rbb$};
			\node at (1.5+\shift,3.35) {\footnotesize $[0,\infty)\times\Rbb$};
			
			\filldraw[color=pythongreen] (-1.25+\shift, 1.25) circle (2pt);
			\filldraw[color=pythongreen] (1.5+\shift, 0.75) circle (2pt);
			\filldraw[color=pythongreen] (-0.5+\shift, 2) circle (2pt);
			
			\draw [thick, -stealth, color=pythongreen] (-3, -0.5) to [out=-30,in=210] (3,-0.5);
			\node[color=pythongreen] at (0, -1) {$\Phi_2$};
		\end{tikzpicture}
	
		\caption{Visualization of the open book and the folding map.
		We show (left) the open book $\B_{3,2}$ with $K=3$ pages and dimension $m=2$.
		Then, we show (center) the folding map $\Phi_2:\B_{3,2}\to\Rbb^2$ which (right) identifies $\mathcal{P}_2$ with $[0,\infty)\times\Rbb$ and identifies both $\mathcal{P}_1$ and $\mathcal{P}_3$ with $(-\infty,0]\times\Rbb$.}
		\label{fig:spider-book}
	\end{figure}
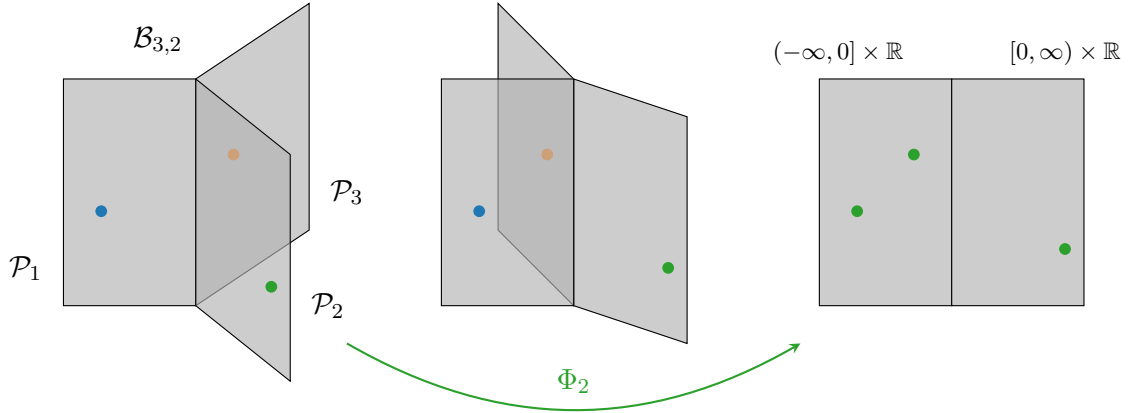

	First, we precisely define the open book.
	Fix integers $K\ge 3$ and $m\ge 1$, and define the half-space $\mathcal{H}=[0,\infty)\times\Rbb^{m-1}$.
	Define the \textit{$K$-page open book of dimension $m$}, denoted $\B_{K,m}$, as the set $\bigsqcup_{1\le j\le K}(\{j\}\times \mathcal{H})$, with the points $\{(j,0,z):1\le j\le K,\ z\in\Rbb^{m-1}\}$ identified for each $z\in\Rbb^{m-1}$.
	The set resulting from identifying $\{(j,0,z):1\le j\le K,\ z\in\Rbb^{m-1}\}$ is called the \textit{spine} and denoted $\mathcal{S}$, and the point resulting from identifying $\{(j,0):1\le j\le K\}$ is called the \textit{origin} and denoted $o$.
	For $1\le j\le K$, the \textit{$j$th (closed) page} of $\B_{K,m}$ is the set $\P_j=\{j\}\times [0,\infty)\times\Rbb^{m-1}$, and the \textit{$j$th open page} is the interior $\P^{\circ}_j=\{j\}\times (0,\infty)\times\Rbb^{m-1}$ of $\P_j$.
	Note that $\B_{K,m}$ equals the disjoint union $\mathcal S\sqcup \bigsqcup_{1\le j\le K}\P_j^{\circ}$.
	Now define $d:\B_{K,m}\times\B_{K,m}\to[0,\infty)$ via
	\begin{equation*}
		d((j,v),(j',v'))=
		\begin{cases}
			\|v-v'\| &\textnormal{ if } j=j',\\
			\|v-Rv'\| &\textnormal{ if } j\neq j',
		\end{cases}
	\end{equation*}
	for all $(j,v),(j',v')\in\B_{K,m}$, where $R:\Rbb^m\to\Rbb^m$ denotes reflection in the first coordinate, i.e., $R(x_1,x_2,\ldots,x_m) = (-x_1,x_2,\ldots,x_m)$.
	It is known \cite[Lemma~1.3]{StickyOpenBooks} that $(\B_{K,m},d)$ is a complete metric space of nonpositive Alexandrov curvature, i.e., a Hadamard space.
	As such, we can apply the theory from \cite{SturmNPC} which provides many useful results for non-Euclidean statistics in Hadamard spaces.
	See the left of Figure~\ref{fig:spider-book} for an illustration of the $3$-page open book.
	Throughout the paper we follow the convention that the symbol $y$ represents an element of $\B_{K,m}$, that $j$ is a page index, and that $v$ is an element of $\mathcal{H}$, e.g., we write $y=(j,v)$.
	When $m=1$, we refer to the open book $\B_{K,1}$ as the \textit{$K$-spider}, the pages $\mathcal{P}_1,\ldots, \mathcal{P}_K$ as its \textit{legs}, and the spine $\S$ as its \textit{head}.
	
	Throughout the paper, $\mu$ denotes a probability measure on $\B_{K,m}$, and $Y_1,Y_2,\ldots$ denotes a sequence of i.i.d. samples from $\mu$, on some probability space $(\Omega,\F,\Pbb_{\mu})$ with expectation denoted $\Ebb_{\mu}$.
	All ``almost sure'' statements are understood to hold with respect to $\Pbb_{\mu}$.
	We write $Y_i=(J_i,V_i)$ for random variables $1\le J_i\le K$ and $V_i\in \mathcal{H}$, and we do not assume that $J_i$ and $V_i$ are independent of each other.
	We also write $Y$ for an additional sample from $\mu$, for convenience in some expressions, and $Y=(J,V)$.
	
	Next, we define Fr\'echet means and some related conditions.
	If we have the first moment condition $\Ebb_{\mu}[d(Y,o)] <\infty$, then a result of \cite{SturmNPC} implies that the \textit{population Fr\'echet mean} defined as
	\begin{equation*}
		F_{\mu} := \underset{x\in \B_{K,m}}{\arg\min}\,\Ebb_{\mu}\left[d^2(x,Y)-d^2(o,Y)\right]
	\end{equation*}
	exists and is unique.
	If we additionally have the stronger second moment condition $\Ebb_{\mu}[d^2(Y,o)] <\infty$, then the population Fr\'echet mean can equivalently be written as
	\begin{equation*}
		F_{\mu} = \underset{x\in \B_{K,m}}{\arg\min}\,\Ebb_{\mu}\left[d^2(x,Y)\right].
	\end{equation*}
	We also define the \textit{sample Fr\'echet mean} as
	\begin{equation*}
		F_n := \underset{x\in \B_{K,m}}{\arg\min}\,\frac{1}{n}\sum_{i=1}^{n}d^2(x,Y_i),
	\end{equation*}
	which exists and is unique for all $n\in\Nbb$, for any $\mu$.
	In fact, it is known \cite{SturmNPC} that the first moment condition implies $\Pbb_{\mu}(F_n\to F_{\mu} \textnormal{ as }n\to\infty) = 1$, which we refer to as the \textit{law of large numbers}.
	We refer to $\{F_n\}_{n\in\Nbb}$ as the \textit{Fr\'echet mean process}.
	
	We say that \textit{stickiness occurs for $\mu$} or that $\mu$ \textit{exhibits stickiness} if we have
	\begin{equation*}
		\Pbb_{\mu}(F_n\in\mathcal S \textnormal{ for sufficiently large }n\in\Nbb)=1.
	\end{equation*}
	To understand this further, we define, for each $1\le j\le K$, the \textit{$j$th folding map} $\Phi_j:\B_{K,m}\to\Rbb^m$ as follows, which we visualize in Figure~\ref{fig:spider-book}:
	\begin{equation*}
		\Phi_j(j',v) = \begin{cases}
			Rv &\textnormal{ if } j'\neq j,\\
			v &\textnormal{ if } j'=j.
		\end{cases}
	\end{equation*}
	A well-known necessary and sufficient condition for stickiness occurring is
	\begin{equation}\label{eqn:stickiness-condition}
		\max_{1\le j\le K}\Ebb_{\mu}\big[(\Phi_j(Y))_1\big] < 0,\tag{S}
	\end{equation}
	where $v_1$ represents the first coordinate of a vector $v\in\Rbb^m$.
	Thus, condition~\eqref{eqn:stickiness-condition} states that the first coordinate of $\Phi_j(Y)$ is negative on average, for all $1\le j\le K$.

	\subsection{Sticking Times and Folded Random Walks}\label{subsec:represent}
	
	Next, we precisely define the sticking time $T$ of the Fr\'echet mean process $\{F_n\}_{n\in\Nbb}$, and we give a representation of the tail event $\{T>n\}$ in terms of a certain collection of random walks.

	\begin{definition}
		The \textit{sticking time} is the random variable
		\begin{equation*}
			T:=\inf\{n\in\Nbb: F_m \in \mathcal S \textnormal{ for all }m\ge n\} =	\sup\{n\in\Nbb: F_n \notin \mathcal S\}+1,
		\end{equation*}
		with the convention that $\inf\varnothing= \infty$.
	\end{definition}

	Observe that $\mu$ exhibits stickiness if and only if $\Pbb_{\mu}(T<\infty)=1$.
	To gain some insight into why analyzing $T$ is difficult, note that, with respect to the filtration generated by $Y_1,Y_2,\ldots$, the random time $T$ is not a stopping time.
		
	Next, define the \textit{$j$th folded random walk} $S_j:=\{S_{j,n}\}_{n\in\Nbb}$ for each $1\le j\le K$ via $S_{j,n}:=\sum_{i=1}^{n}(\Phi_j(Y_i))_1$ for $n\in\Nbb$, which is a random walk in $\Rbb$.
	Also, for each $1\le j\le K$ let us define $\sigma_j:=\sup\{n\ge 0: S_{j,n}> 0\}$ as the last passage time (strictly) above zero.
	
	Now we give the following simple result, which is the key insight of this work.
	Note that it is essentially an immediate consequence of \cite[Lemma~3.3]{StickyOpenBooks}. (See also \cite[Lemma~4]{UlmerVanHuckemann} for the special case $m=1$, and \cite[Theorem~3]{BardenLe} for the more general setting of ``orthant spaces''.)
	
	\begin{lemma}\label{lem:sticking-time-rep}
		We have $T = \max_{1\le j\le K}\sigma_j+1$ almost surely.
	\end{lemma}
	
	\begin{proof}
		We first show that, for every $n\in\Nbb$ and every $1\le j\le K$, we have
		\begin{equation}\label{eq:loc}
			\left\{F_n\in\P^{\circ}_j\right\} =\left\{S_{j,n}>0\right\}.
		\end{equation}
		The first claim of \cite[Lemma~3.3]{StickyOpenBooks} exactly states that the left side is contained in the right side.
		For the converse, suppose $S_{j,n}>0$ so that the second claim of \cite[Lemma~3.3]{StickyOpenBooks} gives $F_n\in\P^{\circ}_j\cup\mathcal S=\P_j$ with $(\Phi_j(F_n))_1=S_{j,n}/n>0$, hence $F_n\in\P^{\circ}_j$.
		This establishes \eqref{eq:loc}, hence also
		\begin{equation}\label{eq:offspine}
			\left\{F_n\notin\mathcal S\right\} =\left\{\max_{1\le j\le K}S_{j,n}>0\right\}
		\end{equation}
		for all $n\in\Nbb$.
		Now use~\eqref{eq:offspine} and the definitions to get
		\begin{equation*}
			\{n\in\Nbb: F_n\notin\mathcal S\} =\left\{n\in\Nbb:\max_{1\le j\le K}S_{j,n}>0\right\} =\bigcup_{1\le j\le K}\{n\in\Nbb:S_{j,n}>0\}.
		\end{equation*}
		Taking supremum over $n\in\Nbb$ finishes the proof.
	\end{proof}

	To understand the value of this result, note that the stickiness condition~\eqref{eqn:stickiness-condition} forces each folded random walk to have negative drift and hence to have a finite last-passage time above zero; consequently, using the probabilistic tools developed in Section~\ref{sec:RW}, we can describe the distribution of these constituent last passage times quantitatively.

	\begin{remark}
		What we call ``sticky'' in this paper is sometimes called ``fully sticky'' by other authors (e.g., \cite[Definition~2.10]{StickyOpenBooks}) who also consider the ``partly sticky'' case where there is some (unique) index $1\le j\le K$ satisfying $\Ebb_{\mu}\big[(\Phi_j(Y))_1\big] = 0$.
		In terms of the folded random walks, the partly sticky case is equivalent to one folded walk being recurrent and all others being transient.
	\end{remark}
	
	Next, we introduce some further integrability of $\mu$ that will be used in order to apply the large deviations and ruin theory results from Section~\ref{sec:RW} to each folded random walk.
	Specifically, we require Cram\'er's condition:
	\begin{equation}\label{eqn:exp-moment-condition}
		\textnormal{ there exists } \theta_j>0 \textnormal{ satisfying } \Ebb_{\mu}\big[e^{\theta_j(\Phi_j(Y))_1}\big]=1\textnormal{ for each } 1\le j\le K. \tag{CC}
	\end{equation}
	Under~\eqref{eqn:stickiness-condition} each such root is unique, and its existence forces $\Pbb_{\mu}(Y\in\P^{\circ}_j)=\Pbb_{\mu}((\Phi_j(Y))_1>0)>0$ for each $1\le j\le K$, so that each folded random walk takes strictly positive values with positive probability.
	
	Lastly, we introduce a non-lattice condition on the steps of the folded random walks, i.e.,
		\begin{equation}\label{eqn:non-lattice-condition}
			(\Phi_j(Y))_1 \textnormal{ has a non-lattice distribution for all } 1 \le j\le K. \tag{NL}
	\end{equation}
	This condition is primarily used to simplify our analyses and the resulting expressions; we treat one specific lattice case in detail in Example~\ref{ex:K-pt}.
	
	\subsection{Precise Asymptotics for the Sticking Time}\label{subsec:asymp}
	
	This subsection contains our first main result, which allows us to analyze the probabilities $\Pbb_{\mu}(T>n)$ as $n\to\infty$.
	As a reminder, $\mu$ is a fixed probability measure on the open book $\B_{K,m}$, and $Y_1,Y_2,\ldots$ are i.i.d. samples from $\mu$.
	
	To state our result, we introduce some notation.
	For each $1\le j\le K$, we let $\nu_j$ be the probability measure on $\Rbb$ denoting the law of $(\Phi_j(Y))_1$ when $Y$ has law $\mu$; in other words, $\nu_j$ is the law of the first coordinate of the data when folded onto page $j$.
	Then, we adopt the notation from Section~\ref{sec:RW}, but we simply use $j$ subscripts, rather than $\nu_j$ subscripts, for simplicity.
	So, the function $\Lambda_j:\Rbb\to \Rbb\cup\{\infty\}$ denotes the CGF of $\nu_j$, the value $\gamma_j>0$ is the minimizer appearing in Cram\'er's theorem~\eqref{eqn:Cramer-thm}, and the function $h_j:\Rbb\to [0,1]$ and the value $\theta_j>0$ are those appearing in the Cram\'er-Lundberg theorem~\eqref{eqn:CR-limit} and Lundberg's bound~\eqref{eqn:Lundberg-bd}.
	
	We also introduce some new notation for the pre-factors appearing in the Bahadur-Rao theorem~\eqref{eqn:BH-thm}.
	Following the previous paragraph, we write $c_j$ for the constant $c_{\nu_j}$ defined in~\eqref{eqn:BH-thm}, and we also define
	\begin{equation}\label{eqn:cprime-nu-def}
		C_{j}:=\frac{1}{(\theta_j-\gamma_{j})\sqrt{2\pi \Lambda_{j}''(\gamma_{j})}},
	\end{equation}
	which satisfies $(\theta_j-\gamma_j)C_j = \gamma_jc_j$, or equivalently $c_j+C_j=\theta_jc_j/(\theta_j-\gamma_j)$.
	(As we will see later in the proofs, it turns out that $C_j$ is the pre-factor appearing in the Bahadur-Rao theorem for a certain random walk whose step distribution is a Cram\'er tilt of $\nu_j$.)
	
	Then we have the following:

	\begin{theorem}\label{thm:sticking-time-tail}
		Under conditions~\eqref{eqn:stickiness-condition},~\eqref{eqn:exp-moment-condition}, and~\eqref{eqn:non-lattice-condition} we have
		\begin{equation}\label{eqn:sticking-time-tail}
			\frac{\max_{j\in\mathcal{J}}c_j}{\sqrt{n}}(1+o(1))e^{-\alpha n}\le \Pbb_{\mu}(T>n) \le \frac{\sum_{j\in\mathcal{J}}(c_j+C_j)}{\sqrt{n}}(1+o(1))e^{-\alpha n}
		\end{equation}
		as $n\to\infty$, where
		\begin{equation*}
			\alpha:=\min_{1\le j\le K}\left(-\Lambda_j(\gamma_j)\right)\qquad\textnormal{and}\qquad \mathcal{J}=\{1\le j\le K: -\Lambda_j(\gamma_j)=\alpha\}.
		\end{equation*}
	\end{theorem}

	\begin{proof}
		The lower bound is easy.
		Fix any $j\in\mathcal{J}$ so that we have $\alpha=-\Lambda_j(\gamma_j)$, and then use Lemma~\ref{lem:sticking-time-rep} and the Bahadur-Rao theorem~\eqref{eqn:BH-thm} to get
		\begin{equation*}
			\Pbb_{\mu}(T>n) \ge \Pbb_{\mu}(\sigma_j\ge n) \ge\Pbb_{\mu}(S_{j,n}>0)= \frac{c_j}{\sqrt{n}}(1+o(1))e^{-\alpha n}
		\end{equation*}
		as $n\to\infty$.
		The upper bound takes a few more steps.
		First, use the Markov property to write
		\begin{equation*}
			\Pbb_{\mu}(\sigma_j\ge n) = \Pbb_{\mu}(S_{j,m}>0 \textnormal{ for some }m\ge n) = \Pbb_{\mu}(S_{j,n}>0)+\Ebb_{\mu}[h_j(S_{j,n})\ind\{S_{j,n}\le 0\}],
		\end{equation*}
		since we have $h_j\equiv 1$ on $(0,\infty)$, and note that the asymptotics of the first term follow from \eqref{eqn:BH-thm} as above.
		For the second term, let $\theta_j>0$ denote the Cram\'er root of $\nu_j$, and define the probability measure $\tilde{\Pbb}_{j}$ whose Radon-Nikodym derivative with respect to $\Pbb_{\mu}$ is  $e^{\theta_j S_{j,n}}$ and write $\tilde{\Ebb}_{j}$ for the corresponding expectation; there is no normalization required since we have $\Lambda_j(\theta_j)=0$ by construction.
		Now, Lundberg's bound~\eqref{eqn:Lundberg-bd} implies
		\begin{align*}
			\Ebb_{\mu}[h_j(S_{j,n})\ind\{S_{j,n}\le 0\}] &\le \Ebb_{\mu}[e^{\theta_jS_{j,n}}\ind\{S_{j,n}\le 0\}] \\
			&= \tilde{\Ebb}_{j}[\ind\{S_{j,n}\le 0\}]  \\
			&= \tilde{\Pbb}_{j}(S_{j,n}\le 0).
		\end{align*}
				Next observe that $e^{\theta_jS_{j,n}}$ factorizes as $\prod_{i=1}^{n}e^{\theta_j(\Phi_j(Y_i))_1}$, so $\{S_{j,n}\}_{n\in\Nbb}$ is still a random walk under $\tilde{\Pbb}_{j}$, whose increments have cumulant generating function $\lambda\mapsto\Lambda_j(\lambda+\theta_j)-\Lambda_j(\theta_j)= \Lambda_j(\lambda+\theta_j)$.
		Hence, under $\tilde{\Pbb}_{j}$, the reflected walk $\{-S_{j,n}\}_{n\in\Nbb}$ has increments with cumulant generating function $\lambda\mapsto\Lambda_j(\theta_j-\lambda)$, which has derivative $-\Lambda_j'(\theta_j)<0$ at $\lambda=0$, unique minimizer $\theta_j-\gamma_j>0$, minimum value $\Lambda_j(\gamma_j)$, and second derivative $\Lambda_j''(\gamma_j)$ at its minimizer; moreover, these increments are non-lattice by~\eqref{eqn:non-lattice-condition}, since exponential tilting and reflection both preserve the property of being non-lattice.
		Therefore, the Bahadur-Rao theorem~\eqref{eqn:BH-thm} applied to $\tilde{\Pbb}_{j}(S_{j,n}\le 0)=\tilde{\Pbb}_{j}(-S_{j,n}\ge 0)$ yields the same exponential rate of decay as for $\Pbb_{\mu}(S_{j,n}> 0)$, but with pre-factor $C_j$ in place of $c_j$.
		Consequently, we have shown
		\begin{equation*}
			\Pbb_{\mu}(\sigma_j\ge n) \le \Pbb_{\mu}(S_{j,n}>0) + \tilde{\Pbb}_{j}(S_{j,n}\le 0) = \frac{c_j+C_j}{\sqrt{n}}(1+o(1))e^{-(-\Lambda_j(\gamma_j))n}
		\end{equation*}
		as $n\to\infty$.
		Lastly, use the union bound to get
		\begin{align*}
			\Pbb_{\mu}(T>n) &\le \sum_{1\le j\le K}\Pbb_{\mu}(\sigma_j\ge n) \\
			&= \sum_{1\le j\le K} \frac{c_j+C_j}{\sqrt{n}}(1+o(1))e^{-(-\Lambda_j(\gamma_j))n} \\
			&= \sum_{j\in\mathcal{J}} \frac{c_j+C_j}{\sqrt{n}}(1+o(1))e^{-\alpha n}
		\end{align*}
		where we used that all terms with index $j\notin\mathcal{J}$ have faster exponential decay than those with index $j\in\mathcal{J}$, hence are negligible.
		This completes the proof.
	\end{proof}

	The preceding result gives a precise estimate on $\Pbb_{\mu}(T>n)$ in the sense that both sides decay like $n^{-1/2}e^{-\alpha n}$.
	Next we give some concrete examples in which we can directly compute all of the values appearing in this result.
	Since we always reduce to the first component of the random walk along each page, it suffices to focus on the $m=1$ case of the $K$-spider.
			
		\begin{example}[Gamma radial distribution]\label{ex:radial}
	Consider the setting that the page index $J$ is sampled uniformly among $\{1,\ldots, K\}$ and that $V$ is sampled from distribution $\rho$, independent of $J$.
	We refer to such a distribution $\mu$ as a ``radial'' distribution, since it depends only on the distance from the spine.
	For concreteness, consider the case that $\rho$ is a Gamma distribution with shape $m$ and rate $m$, for $m\ge 1$.
	Then, $\Ebb_{\mu}[V]=1$, so condition~\eqref{eqn:stickiness-condition} holds because $K\ge 3$, and condition~\eqref{eqn:non-lattice-condition} holds because each folded distribution is absolutely continuous.

	Next we compute $\alpha$, $\mathcal{J}$, and the constants $c_j$ and $C_j$ appearing in Theorem~\ref{thm:sticking-time-tail}.
	To do this, note by symmetry that all of the functions $\Lambda_j$ are equal, so we may write them as $\Lambda$ and note
	\begin{equation*}
		\Lambda(\lambda) =\log\left(\frac1K\left(1-\frac{\lambda}{m}\right)^{-m}+\frac{K-1}{K}\left(1+\frac{\lambda}{m}\right)^{-m}\right)
	\end{equation*}
	for $|\lambda|<m$, and $\Lambda(\lambda)\to\infty$ as $\lambda\uparrow m$.
	Differentiating this and setting the result to zero leads to the minimizer
	\begin{equation*}
		\gamma=m\cdot\frac{t-1}{t+1},
	\end{equation*}
	where $t:=(K-1)^{1/(m+1)}$, as well as the minimum value
	\begin{equation*}
		\Lambda(\gamma)=\log\left(\frac1K\left(\frac{t+1}{2}\right)^m+\frac{K-1}{K}\left(\frac{t+1}{2t}\right)^m\right)=
		\log\left(\frac{(t+1)^{m+1}}{2^m K}\right).
	\end{equation*}
	By symmetry we have $\mathcal{J}=\{1,\ldots,K\}$, and
	\begin{equation*}
		\alpha=-\Lambda(\gamma)=\log\!\left(\frac{2^m K}{(t+1)^{m+1}}\right).
	\end{equation*}
	Next, we can directly calculate
	\begin{equation*}
		\Lambda''(\gamma)
		=
		\frac{m+1}{m}\cdot\frac{(t+1)^2}{4t},
	\end{equation*}
	hence, for every $j\in\mathcal{J}$, we have $c_j=c^{\ast}$ where
	\begin{equation*}
		c^{\ast}:=\frac{1}{\gamma\sqrt{2\pi\Lambda''(\gamma)}}=\frac{\sqrt{2t}}{(t-1)\sqrt{\pi m(m+1)}}.
	\end{equation*}
	Similarly, for every $j\in\mathcal{J}$ we have $C_j=C^{\ast}$ where
	\begin{equation*}
		C^{\ast}:=\frac{1}{(\theta-\gamma)\sqrt{2\pi\Lambda''(\gamma)}}=\frac{\gamma}{\theta-\gamma}\,c^{\ast},
	\end{equation*}
	and where $\theta$ is the unique root of $\Lambda$ in $(0,m)$, which exists because $\Lambda(\gamma)<0$ and $\Lambda(\lambda)\to\infty$ as $\lambda\uparrow m$.
	Since $\#\mathcal{J}=K$, Theorem~\ref{thm:sticking-time-tail} therefore gives
	\begin{equation}\label{eqn:radial-bounds}
		\frac{c^{\ast}}{\sqrt{n}}(1+o(1))\left(\frac{(t+1)^{m+1}}{2^mK}\right)^{n}\le \Pbb_{\mu}(T>n) \le \frac{K\theta c^{\ast}}{(\theta-\gamma)\sqrt{n}}(1+o(1))\left(\frac{(t+1)^{m+1}}{2^mK}\right)^{n}
	\end{equation}
	as $n\to\infty$.

	For general $m\ge 1$ the root $\theta$ has no closed form; in the exponential case $m=1$, everything can be made explicit.
	Indeed, then $t=\sqrt{K-1}$, and the equation $\Lambda(\theta)=0$ reads $1/(1-\theta) + (K-1)/(1+\theta)=K$ whose unique positive root is $\theta=(K-2)/K=(t^2-1)/(t^2+1)$.
	Combining this with $\gamma=(t-1)/(t+1)$ gives $\theta/(\theta-\gamma)=(t+1)^2/(2t)$, and therefore~\eqref{eqn:radial-bounds} becomes
	\begin{equation*}
		\frac{\sqrt{t}}{(t-1)\sqrt{\pi n}}(1+o(1))\left(\frac{(t+1)^{2}}{2K}\right)^{n}\le \Pbb_{\mu}(T>n) \le \frac{K(t+1)^2}{2(t-1)\sqrt{\pi t n}}(1+o(1))\left(\frac{(t+1)^{2}}{2K}\right)^{n}
	\end{equation*}
	as $n\to\infty$.
	For example, when $K=3$ this reads
	\begin{equation*}
		\frac{1.6198}{\sqrt{n}}(1+o(1))e^{-0.0290n}\le \Pbb_{\mu}(T>n) \le \frac{10.0135}{\sqrt{n}}(1+o(1))e^{-0.0290n}
	\end{equation*}
	as $n\to\infty$.
	\end{example}

	\begin{example}[heterocrural\footnote{\textit{crural} (adj): of or pertaining to legs, in anatomy} model]
	Consider the case that the page index $J$ is selected uniformly at random among $\{1,\ldots, K\}$ and, conditional on $\{J=j\}$, the radius $V$ is sampled from a probability measure $\rho_j$ on $(0,\infty)$; here, $\rho_1,\ldots, \rho_K$ is a fixed collection of probability measures, which need not be equal to each other.
	Note that condition \eqref{eqn:stickiness-condition} is equivalent to $\int_{0}^{\infty}v\diff \rho_j(v)< \sum_{j'\neq j}\int_{0}^{\infty}v\diff \rho_{j'}(v)$ for all $1\le j\le K$, which holds, for instance, if all $\rho_j$ have the same expectation.
	Note that condition \eqref{eqn:exp-moment-condition} is equivalent to the usual Cram\'er condition holding for all $\rho_1,\ldots, \rho_K$.
	
	To be concrete, we fix $K=3$ and we take some particular choices of $\rho_1,\rho_2,\rho_3$ for which we can do numerical calculations, although we do not have closed-form expressions for the resulting probabilities.
	More precisely, take $\rho_1=\delta_{1}, \rho_2 = \textnormal{Exp}(1)$, and $\rho_3= |\mathcal{N}(0,\pi/2)|$, where the scaling of the third distribution is used to ensure that they all have expectation equal to 1.
	In other words, the distribution $\mu$ has three different types of legs: a degenerate leg, an exponential leg, and a half-Gaussian leg.
	
		In order to compute $c_j$ and $C_j$, we note that each of $(\Phi_1(Y))_1,(\Phi_2(Y))_1$, and $(\Phi_3(Y))_1$ is non-lattice, so we directly compute
	\begin{equation*}
		c_1 \approx 1.4244, \qquad c_2 \approx 1.4161, \qquad\textnormal{and}\qquad c_3 \approx 1.4135.
	\end{equation*}
	(We note that all three of $c_1,c_2,c_3$ happen to be quite close to each other, despite the three legs having rather different distributions.)
	We also numerically solve for the Cram\'er roots
	\begin{equation*}
		\theta_1 \approx 0.5493, \qquad \theta_2 \approx 0.3748, \qquad\textnormal{and}\qquad \theta_3 \approx 0.4495,
	\end{equation*}
	which by~\eqref{eqn:cprime-nu-def} give
	\begin{equation*}
		C_1 \approx 1.3161, \qquad C_2 \approx 1.6084, \qquad\textnormal{and}\qquad C_3 \approx 1.4618.
	\end{equation*}
	In particular, since $\mathcal{J}=\{2\}$ is a singleton, we have $\max_{j\in\mathcal{J}}c_j=c_2\approx1.4161$ and $\sum_{j\in\mathcal{J}}(c_j+C_j)=c_2+C_2\approx 3.0245$.

	Putting it all together, Theorem~\ref{thm:sticking-time-tail} gives
	\begin{equation*}
		\frac{1.4161}{\sqrt{n}}(1+o(1))e^{-0.0351 n}
		\leq
		\Pbb_{\mu}(T>n)
		\leq
		\frac{3.0245}{\sqrt{n}}(1+o(1))e^{-0.0351 n}
	\end{equation*}
	as $n \to \infty$, where the pre-factor and exponent reflect the fact that the tail of the sticking time is governed by the exponential leg.
\end{example}

While our non-lattice condition~\eqref{eqn:non-lattice-condition} excludes lattice distributions, such as those supported on finitely many equally spaced points, the following example is a particular lattice setting where we can reach a similar conclusion by more elementary methods.
In fact, our conclusion in the following example is sharper than that of Theorem~\ref{thm:sticking-time-tail}, since we determine the exact pre-factor (i.e., exact asymptotics) for the probability $\Pbb_{\mu}(T>n)$.
	
	\begin{example}[$K$-Rademacher distribution]\label{ex:K-pt}
	Consider the case that the page index $J$ is sampled from some distribution $w\in\Delta_K$ and the radius $V$ is taken to be 1 almost surely; we refer to the resulting $\mu\in\mathcal{P}(\B_{K,1})$ as the \textit{$K$-Rademacher distribution}, and we write $w^{\ast}:=\max_{1\le j\le K}w_j$ throughout.
	We note in this case that condition~\eqref{eqn:exp-moment-condition} is trivially satisfied whenever $w_j>0$ for all $1\le j\le K$, and that condition~\eqref{eqn:stickiness-condition} holds if and only if $w_j<1/2$ for all $1\le j\le K$.
	Each folded random walk is now a simple biased random walk, namely $S_{j,n}=2\sum_{i=1}^{n}\ind\{J_i=j\}-n$ where $\sum_{i=1}^{n}\ind\{J_i=j\}$ has a binomial distribution with parameters $n$ and $w_j$.
	Thus, we will use a similar approach to the proof of Theorem~\ref{thm:sticking-time-tail}, but with Stirling's approximation in place of the Bahadur-Rao theorem.
	
	First we follow the beginning of the proof of Theorem~\ref{thm:sticking-time-tail} but with an exact expression replacing Lundberg's bound.
	That is, recall the well-known identity	
	\begin{equation*}
		h_j(x) = \left(\frac{w_j}{1-w_j}\right)^{1-x}
	\end{equation*}
	for $x\le0$, which is an exact form of Lundberg's bound~\eqref{eqn:Lundberg-bd} up to an additional factor of $w_j/(1-w_j)$.
	Then we compute
	\begin{align*}
		\Pbb_{\mu}(\sigma_j\ge n) &= \Pbb_{\mu}(S_{j,n}>0)+\Ebb_{\mu}\big[h_j(S_{j,n})\ind\{S_{j,n}\le 0\}\big] \\
		&= \Pbb_{\mu}(S_{j,n}>0)+\frac{w_j}{1-w_j}\,\Ebb_{\mu}\left[\left(\frac{w_j}{1-w_j}\right)^{-S_{j,n}}\ind\{S_{j,n}\le 0\}\right]
	\end{align*}
	for all $n\in\Nbb$.
	Moreover, tilting the law $\nu_j$ by the Cram\'er root $\theta_j=\log((1-w_j)/w_j)$ simply interchanges the weights $w_j$ and $1-w_j$, hence replaces $S_{j,n}$ by $-S_{j,n}$ in distribution, and therefore
	\begin{equation}\label{eqn:K-pt-exact}
		\Pbb_{\mu}(\sigma_j\ge n) = \Pbb_{\mu}(S_{j,n}>0)+\frac{w_j}{1-w_j}\,\Pbb_{\mu}(S_{j,n}\ge 0)
	\end{equation}
	for all $n\in\Nbb$.
	
	Second, we compute the asymptotics of the two probabilities appearing in~\eqref{eqn:K-pt-exact}, which requires us to separately consider the case of even and odd $n$.
	Since $\nu_j$ is lattice with span $2$, the lattice case of \cite[Theorem~3.7.4]{LargeDeviations} applies at the point $0$, which lies in the support of $S_{j,n}$ exactly when $n$ is even; thus, we have the following when $n\to\infty$ along even integers:
	\begin{equation}\label{eqn:K-pt-even}
		\Pbb_{\mu}(S_{j,n}\ge0) = \sqrt{\frac{2}{\pi n}}\cdot\frac{1-w_j}{1-2w_j}\,\big(4w_j(1-w_j)\big)^{\frac{n}{2}}(1+o(1))
	\end{equation}
	Also, Stirling's formula gives the following, also for even $n$:
	\begin{equation}\label{eqn:K-pt-atom}
		\Pbb_{\mu}(S_{j,n}=0) = \binom{n}{n/2}\big(w_j(1-w_j)\big)^{\frac{n}{2}} = \sqrt{\frac{2}{\pi n}}\,\big(4w_j(1-w_j)\big)^{\frac{n}{2}}(1+o(1)).
	\end{equation}
	For the odd case, the events $\{S_{j,n}>0\}$ and $\{S_{j,n}\ge0\}$ coincide, hence conditioning on the last increment gives
	\begin{equation*}
		\Pbb_{\mu}(S_{j,n+1}>0) = w_j\Pbb_{\mu}(S_{j,n}\ge 0)+(1-w_j)\Pbb_{\mu}(S_{j,n}\ge 2) = \Pbb_{\mu}(S_{j,n}\ge0)-(1-w_j)\Pbb_{\mu}(S_{j,n}=0).
	\end{equation*}
	Substituting~\eqref{eqn:K-pt-even} and~\eqref{eqn:K-pt-atom} into these identities and then into~\eqref{eqn:K-pt-exact}, and using $$\big(4w_j(1-w_j)\big)^{\frac{n}{2}}=\big(4w_j(1-w_j)\big)^{\frac{n+1}{2}}\big(4w_j(1-w_j)\big)^{-\frac{1}{2}}$$ together with $\sqrt{n}/\sqrt{n+1}\to1$ in order to express the odd case in terms of $n+1$, we obtain
	\begin{equation*}
	\Pbb_{\mu}(\sigma_j\ge n) = \frac{c_{j,n}}{\sqrt{n}}(1+o(1))\,\big(4w_j(1-w_j)\big)^{\frac{n}{2}},
	\end{equation*}
	where
	\begin{equation}\label{eqn:K-pt-lastpassage}
		c_{j,n}:=\sqrt{\frac{2}{\pi}}\cdot\frac{1}{1-2w_j}\times
		\begin{cases}
			2w_j &\textnormal{ if } n \textnormal{ is even},\\
			\sqrt{\frac{w_j}{1-w_j}} &\textnormal{ if } n \textnormal{ is odd},
		\end{cases}
	\end{equation}
	for all $n\in\Nbb$.
		
	Lastly, we pass from the individual last passage times to the sticking time.
	To do this, write $\mathcal{J}=\{1\le j\le K: w_j=w^{\ast}\}$ as well as
	\begin{equation}\label{eqn:K-pt-lastpassage--star}
		c_{n}^{\ast}:=\sqrt{\frac{2}{\pi}}\cdot\frac{1}{1-2w^{\ast}}\times
		\begin{cases}
			2w^{\ast} &\textnormal{ if } n \textnormal{ is even},\\
			\sqrt{\frac{w^{\ast}}{1-w^{\ast}}} &\textnormal{ if } n \textnormal{ is odd},
		\end{cases}
	\end{equation}
	for all $n\in\Nbb$.
	Then note that, by the union bound and~\eqref{eqn:K-pt-lastpassage}, we have the following as $n\to\infty$:
	\begin{align*}
		\Pbb_{\mu}(T>n)&\le \sum_{1\le j\le K}\Pbb_{\mu}(\sigma_j\ge n) \\
		&= \sum_{1\le j\le K}\frac{c_{j,n}}{\sqrt{n}}(1+o(1))\,\big(4w_j(1-w_j)\big)^{\frac{n}{2}} \\
		&= \sum_{j\in\mathcal{J}}\frac{c_{j,n}}{\sqrt{n}}(1+o(1))\,\big(4w_j(1-w_j)\big)^{\frac{n}{2}} \\
		&= \frac{c_{n}^{\ast}\#\mathcal{J}}{\sqrt{n}}(1+o(1))\,\big(4w^{\ast}(1-w^{\ast})\big)^{\frac{n}{2}},
	\end{align*}
	since all terms with $j\notin\mathcal{J}$ have a strictly faster exponential rate of decay.
	For a matching lower bound, observe that $\mathcal{N}:=\{\ind\{J_i=j\}: 1\le j\le K,\ i\in\Nbb\}$ is a negatively associated collection of random variables (see \cite{JoagDevProschan}) and that, for each fixed $j$, the event $\{\sigma_j\ge n\}$ is an increasing function of the variables $\{\ind\{J_i=j\}: i\in\Nbb\}$; thus, we have
	\begin{equation*}
		\Pbb_{\mu}(\sigma_j\ge n,\ \sigma_{j'}\ge n)\le \Pbb_{\mu}(\sigma_j\ge n)\,\Pbb_{\mu}(\sigma_{j'}\ge n)
	\end{equation*}
	for all $1\le j\neq j'\le K$.
	Combining this with the union bound and with Lemma~\ref{lem:sticking-time-rep}, we get
	\begin{align*}
		\Pbb_{\mu}(T>n)&\ge \sum_{j\in\mathcal{J}}\Pbb_{\mu}(\sigma_j\ge n)-\sum_{\substack{j,j'\in\mathcal{J}\\ j<j'}}\Pbb_{\mu}(\sigma_j\ge n,\ \sigma_{j'}\ge n) \\
		&\ge \sum_{j\in\mathcal{J}}\Pbb_{\mu}(\sigma_j\ge n)-\sum_{\substack{j,j'\in\mathcal{J}\\ j<j'}}\Pbb_{\mu}(\sigma_j\ge n)\,\Pbb_{\mu}(\sigma_{j'}\ge n).
	\end{align*}
	The first sum on the right side is asymptotically equivalent to the upper bound, and the second sum has a strictly faster exponential rate of decay because of the product structure, hence it vanishes.
	
	Putting it all together, we have shown
	\begin{equation*}
		\Pbb_{\mu}(T>n)= \frac{c_{n}^{\ast}\#\mathcal{J}}{\sqrt{n}}(1+o(1))\,\big(4w^{\ast}(1-w^{\ast})\big)^{\frac{n}{2}}
	\end{equation*}
	as $n\to\infty$, where $c_n^{\ast}$ is defined in~\eqref{eqn:K-pt-lastpassage--star}.
	We note that this has a similar form to the conclusion of Theorem~\ref{thm:sticking-time-tail}, except that the upper and lower bounds match, and that the constant pre-factor requires some correction for lattice effects.
\end{example}

	\subsection{Nonparametric Estimation of Future Unsticking}\label{subsec:estim}
	
	This subsection contains our second main result, in which we provide a fully nonparametric procedure for estimating the probability $\Pbb_{\mu}(T>n\,|\,Y_1,\ldots,Y_n)$ when the distribution $\mu$ of $Y_1,Y_2,\ldots$ is unknown, as $n\to\infty$.
	In other words, we aim to estimate the probability that sticking has not yet occurred conditional on the observed data so far.
	We give further motivation for (and applications of) this problem in Section~\ref{sec:BHV}.
	
	In order to construct our estimator, recall from Lemma~\ref{lem:sticking-time-rep} and the union bound that we have
	\begin{equation*}
	\max_{1\le j\le K}\Pbb_{\mu}(\sigma_j\ge n\,|\,Y_1,\ldots, Y_n)\le\Pbb_{\mu}(T>n\,|\,Y_1,\ldots, Y_n)\le \sum_{1\le j\le K}\Pbb_{\mu}(\sigma_j\ge n\,|\,Y_1,\ldots, Y_n),
	\end{equation*}
	and recall also $\Pbb_{\mu}(\sigma_j\ge n\,|\,Y_1,\ldots,Y_n)=h_j(S_{j,n})$, where $h_j(x):=\Pbb_{\mu}(\tau<\infty\,|\,S_{j,0}=x)$ is the probability of returning to the positive half-line, for a RW started at $x\in\Rbb$ whose steps have distribution $\nu_j$.
	Since $S_{j,n}\to-\infty$ almost surely as $n\to\infty$ for all $1\le j\le K$, we may use the Cram\'er-Lundberg theory \eqref{eqn:CR-limit} to approximate $h_j(S_{j,n})\approx e^{\theta_jS_{j,n}}$, hence the estimation of $\Pbb_{\mu}(T>n\,|\,Y_1,\ldots,Y_n)$ depends only on the ability to estimate the Cram\'er roots $\theta_1,\ldots,\theta_K$, which we can do straightforwardly via $Z$-estimation.
	
	An explicit procedure implementing this idea is given in Algorithm~\ref{alg:unstick}, about which we make the following remarks that follow from basic properties of convexity of cumulant generating functions:
		\begin{itemize}
			\item We always have $\textsc{UnstickingProbability}(Y_1,\ldots,Y_n)\in [0,1]$.
			To see this, recall that $F_n\in\S$ implies $\sum_{i=1}^{n}(\Phi_j(Y_i))_1\le 0$ for all $j$.
			Then note that $\sum_{i=1}^{n}(\Phi_j(Y_i))_1\le 0$ implies $\hat{\theta}_j\ge 0$ hence $\hat{\theta}_j\sum_{i=1}^{n}(\Phi_j(Y_i))_1\le 0$. 
			\item If $F_n\notin\S$, then we have  $\textsc{UnstickingProbability}(Y_1,\ldots,Y_n)=1$.
			To see this, recall that $F_n\notin\S$ implies $\sum_{i=1}^{n}(\Phi_j(Y_i))_1> 0$ for exactly one $j$.
			Then note that $\sum_{i=1}^{n}(\Phi_j(Y_i))_1>0$ implies $\hat{\theta}_j=0$ hence $\hat{\theta}_j\sum_{i=1}^{n}(\Phi_j(Y_i))_1= 0$.
		\end{itemize}
		In other words, $\textsc{UnstickingProbability}$ returns a probability with the right qualitative properties.
		(We will discuss the modified procedure $\textsc{UnstickingProbabilityBootstrap}$ below.)

		\begin{algorithm}[h]
			\caption{Procedures for estimating the unsticking probability for Fr\'echet means in open books.
				The first procedure is downwardly-biased because of the maximum, and the second procedure removes this bias with a bootstrap correction.}\label{alg:unstick}
			\begin{algorithmic}[1]
				
				\Procedure{\textsc{UnstickingProbability}}{$Y_1,\ldots, Y_n$}
				
				\State \textbf{input:} samples $Y_1,\ldots, Y_n$ on $\B_{K,m}$
				\State \textbf{output:} estimate of $\Pbb_{\mu}(T>n\,|\,Y_1,\ldots,Y_n)$
				
				\State \,
				
				\For{$j=1,\ldots, K$}
				\State $\hat{\theta}_j\leftarrow \max\{\lambda\in\Rbb: \frac{1}{n}\sum_{i=1}^{n}\exp(\lambda(\Phi_j(Y_i))_1)=1\}$
				\EndFor
				
				\State \textbf{return} $\exp\big(\max_{1\le j\le K}\hat{\theta}_j\sum_{i=1}^{n}(\Phi_{j}(Y_i))_1\big)$
				\EndProcedure
				
				\State \,
				\State \,
				
				\Procedure{\textsc{UnstickingProbabilityBootstrap}}{$Y_1,\ldots,Y_n$; $B$}
				
				\State \textbf{input:} samples $Y_1,\ldots,Y_n$ on $\B_{K,m}$ and number of replications $B$
				\State \textbf{output:} bias-corrected estimate of $\Pbb_{\mu}(T>n\,|\,Y_1,\ldots,Y_n)$
				
				\State \,
				
				\State $\hat p\leftarrow$ \textsc{UnstickingProbability}$(Y_1,\ldots,Y_n)$
				\State $\hat g\leftarrow \tfrac{1}{n}\log\hat p$
				
				\For{$b=1,\ldots,B$}
				\State $Y_1^{(b)},\ldots,Y_n^{(b)}\leftarrow$ sample $n$ points with replacement from $Y_1,\ldots,Y_n$
				\State $g_b\leftarrow \tfrac{1}{n}\log \textsc{UnstickingProbability}(Y_1^{(b)},\ldots,Y_n^{(b)})$
				\EndFor
				
				\State $\widehat{\mathrm{bias}}\leftarrow n(\tfrac{1}{B}\sum_{b=1}^{B} g_b - \hat g)$
				
				\State \textbf{return} $\min(\exp(\log \hat{p}-\widehat{\mathrm{bias}}),\,1)$
				\EndProcedure
			\end{algorithmic}
		\end{algorithm}
	
		Next we give our main consistency result for Algorithm~\ref{alg:unstick}.
		Before we state it, let us comment on the particular notion of consistency we use.
		First, note that we consider the log probability $\log \Pbb_{\mu}(T>n\,|\,Y_1,\ldots,Y_n)$ rather than the probability $\Pbb_{\mu}(T>n\,|\,Y_1,\ldots,Y_n)$ itself; this is because the probability of interest is exponentially small, and indeed our asymptotic theory in the setting where $\mu$ is known already does not identify the exact constant pre-factor.
		Second, note that we consider the relative log-probability rather than the absolute log-probability; this is because the log-probabilities are of order $\Theta(n)$, so this is similar to considering the normalized log-probabilities $n^{-1}\log \Pbb_{\mu}(T>n\,|\,Y_1,\ldots, Y_n)$.
		
		We also need to slightly strengthen condition~\eqref{eqn:exp-moment-condition}, namely by assuming
		\begin{equation}\label{eqn:double-moment-condition}
			\textnormal{for each } 1\le j\le K \textnormal{ we have } \Ebb_{\mu}\big[e^{\lambda(\Phi_j(Y))_1}\big]<\infty\textnormal{ for some }\lambda>2\theta_j \tag{CC$'$}
		\end{equation}
		which will allow us to control the variance of certain exponential expressions.
		Then, we have the following:
			
	\begin{proposition}\label{prop:alg-consistency}
		Under conditions~\eqref{eqn:stickiness-condition}, \eqref{eqn:double-moment-condition}, and~\eqref{eqn:non-lattice-condition}, we have
		\begin{equation*}
			\left|\frac{\log \textsc{UnstickingProbability}(Y_1,\ldots, Y_n)-\log \Pbb_{\mu}(T>n\,|\,Y_1,\ldots, Y_n)}{\log \Pbb_{\mu}(T>n\,|\,Y_1,\ldots, Y_n)}\right| = O_{\Pbb}(n^{-1/2})
		\end{equation*}
		as $n\to\infty$.
	\end{proposition}
	
	\begin{proof}
		For the sake of brevity throughout, we write $p_n:=\Pbb_{\mu}(T>n\mid Y_1,\dots,Y_n)$ and $\hat p_n:=\textsc{UnstickingProbability}(Y_1,\dots,Y_n)=\exp\!\big(\max_{1\le j\le K}\hat\theta_j S_{j,n}\big)$ so that the claim is equivalent to $|(\log\hat p_n-\log p_n)/\log p_n|=O_{\Pbb}(n^{-1/2})$.
		
		The first step is to control $\log p_n$ in terms of the Cram\'er roots $\theta_1,\ldots, \theta_K$.
		By Lemma~\ref{lem:sticking-time-rep}, the union bound, and the definition of $h_j(S_{j,n})=\Pbb_{\mu}(\sigma_j\ge n\mid Y_1,\dots,Y_n)$, we have
		\begin{equation}\label{eq:prop-sandwich}
			\max_{1\le j\le K}\log h_j(S_{j,n})\ \le\ \log p_n\ \le\ \max_{1\le j\le K}\log h_j(S_{j,n})+\log K.
		\end{equation}
		Under condition~\eqref{eqn:stickiness-condition} each folded walk has negative drift, hence $S_{j,n}\to-\infty$ almost surely by the strong law of large numbers; consequently, the Cram\'er-Lundberg asymptotic~\eqref{eqn:CR-limit} and equation~\eqref{eq:prop-sandwich} together imply
		\begin{equation}\label{eqn:pn-Thetan}
			\frac1n\log p_n\to -\min_{1\le j\le K}\theta_j\big(-\Lambda_j'(0)\big)<0,
		\end{equation}
		almost surely, hence $\log p_n=\Theta_{\Pbb}(n)$.
		Moreover, Lundberg's bound~\eqref{eqn:Lundberg-bd} together with the sharpened Cram\'er-Lundberg approximation \cite{CramerLundberg}, which asserts that $h_j(x)e^{-\theta_j x}$ stays bounded as $x\to-\infty$, guarantees that there exist $A_j>0$ and $x_0<0$ such that $x\le x_0$ implies
		\begin{equation*}
			\max_{1\le j\le K}\big|\log h_j(x)-\theta_j x\big|\le \max_{1\le j\le K}A_j.
		\end{equation*}
		Since $S_{j,n}\to-\infty$ almost surely by the strong law of large numbers, this implies
		\begin{equation}\label{eq:prop-lund}
			\max_{1\le j\le K}\big|\log h_j(S_{j,n})-\theta_j S_{j,n}\big|\le \max_{1\le j\le K}A_j
		\end{equation}
		for sufficiently large $n\in\Nbb$ almost surely.

				Next, we establish $\sqrt n$-consistency of the estimators $\hat{\theta}_1,\ldots,\hat{\theta}_K$ of $\theta_1,\ldots,\theta_K$.
		Fix $1\le j\le K$, let $\lambda_j^{\ast}>2\theta_j$ satisfy $\Ebb_{\mu}\big[e^{\lambda_j^{\ast}(\Phi_j(Y))_1}\big]<\infty$ as in condition~\eqref{eqn:double-moment-condition}, and note that $\hat\theta_j$ is a $Z$-estimator solving $M_{j,n}(\lambda):=\frac1n\sum_{i=1}^{n}\psi_{j,\lambda}(Y_i)=0$ for $\psi_{j,\lambda}(y):=e^{\lambda(\Phi_j(y))_1}-1$, whose population counterpart is $M_j(\lambda) := \Ebb_{\mu}\big[\psi_{j,\lambda}(Y)\big]=e^{\Lambda_j(\lambda)}-1$.
		Since $e^{\lambda x}\le 1+e^{\lambda_j^{\ast}x}$ for all $\lambda\in[0,\lambda_j^{\ast}]$ and $x\in\Rbb$, the function $\Lambda_j$ is finite on $[0,\lambda_j^{\ast}]$, hence smooth on $(0,\lambda_j^{\ast})$.
		To show the desired result, we just need to verify the hypotheses of \cite[Theorem~5.21]{vanDerVaart}.
		First, note that we have $\hat\theta_j\to\theta_j$ almost surely from the fact that $M_{j}$ is strictly convex.
		Second, $\lambda\mapsto\psi_{j,\lambda}(y)$ is Lipschitz near $\theta_j$ with a square-integrable Lipschitz constant.
		Indeed, fix $\varepsilon\in(0,\theta_j)$ with $\eta:=\lambda_j^{\ast}-2(\theta_j+\varepsilon)>0$, and write $x:=(\Phi_j(y))_1$ and $x_+:=\max\{x,0\}$.
		Since $\partial_{\lambda}\psi_{j,\lambda}(y)=xe^{\lambda x}$ and $\sup_{x\le 0}|x|e^{ax}=1/(ea)$ for $a>0$, we get
		\begin{equation*}
			\sup_{|\lambda-\theta_j|\le\varepsilon}\big|\partial_{\lambda}\psi_{j,\lambda}(y)\big|\le L_j(y):=\frac{1}{e(\theta_j-\varepsilon)}+x_+e^{(\theta_j+\varepsilon)x_+},
		\end{equation*}
		so the mean value theorem gives $|\psi_{j,\lambda}(y)-\psi_{j,\lambda'}(y)|\le L_j(y)|\lambda-\lambda'|$ for all $\lambda,\lambda'\in[\theta_j-\varepsilon,\theta_j+\varepsilon]$.
		Moreover, $\sup_{x\ge0}x^2e^{-\eta x}=4/(e\eta)^2$ implies $x_+^2e^{2(\theta_j+\varepsilon)x_+}\le 4(e\eta)^{-2}e^{\lambda_j^{\ast}x_+}\le 4(e\eta)^{-2}(1+e^{\lambda_j^{\ast}x})$, hence $L_j\in L^2(\Pbb_{\mu})$ by condition~\eqref{eqn:double-moment-condition}.
						
		Now we put all the pieces together.
		By the triangle inequality and \eqref{eq:prop-sandwich} we can make the following bound, which we analyze via \eqref{eqn:pn-Thetan} and \eqref{eq:prop-lund}:
		\begin{align*}
			&\frac{\big|\log \hat{p}_n -\log p_n\big|}{|\log p_n|} \\
			&\qquad\le \frac{\max_{1\le j\le K}\big|\hat{\theta}_j -\theta_j\big|\max_{1\le j\le K}|S_{j,n}|}{|\log p_n|}
			 + \frac{\max_{1\le j\le K}\big|\theta_j S_{j,n}-\log h_j(S_{j,n})\big|}{|\log p_n|} + \frac{\log K}{|\log p_n|}\\
			 &\qquad= \frac{O_{\Pbb}(n^{-1/2})\,\Theta_{\Pbb}(n)}{\Theta_{\Pbb}(n)}
			 + \frac{O_{\Pbb}(1)}{\Theta_{\Pbb}(n)} + \frac{O_{\Pbb}(1)}{\Theta_{\Pbb}(n)}.
		\end{align*}
		The right side is $O_{\Pbb}(n^{-1/2})$ as desired, and this finishes the proof.
	\end{proof}
	
	While this result establishes an asymptotic guarantee for the relative error of the procedure $\textsc{UnstickingProbability}(Y_1,\ldots, Y_n)$ with respect to $\Pbb_{\mu}(T>n\,|\,Y_1,\ldots,Y_n)$ as $n\to\infty$, there is typically substantial negative bias in practice because of the maximum appearing in Algorithm~\ref{alg:unstick}.
	A standard way to correct for this is by using the bootstrap to estimate this bias and then subtract it from the resulting estimate, and this leads us to the additional procedure $\textsc{UnstickingProbabilityBootstrap}(Y_1,\ldots, Y_n)$ in Algorithm~\ref{alg:unstick}.
	We do not provide additional theory for the bias-corrected procedure, but we illustrate it with a basic Monte Carlo simulation:
	for each sample size $n\in\{10^2,10^3,10^4\}$, we consider $B=300$ trials of Example~\ref{ex:radial} (i.e., a radial exponential distribution on the 3-spider $\B_{3,1}$) for which we can exactly calculate $p_n:=\Pbb_{\mu}(T>n\,|\,Y_1,\ldots, Y_n)$.
	We show the results in Figure~\ref{fig:bootstrap} which demonstrate that bootstrap bias-correction indeed decreases the bias at the cost of slightly increasing the variance.
	
	\begin{figure}[t]
		\centering
		\includegraphics[scale=0.5]{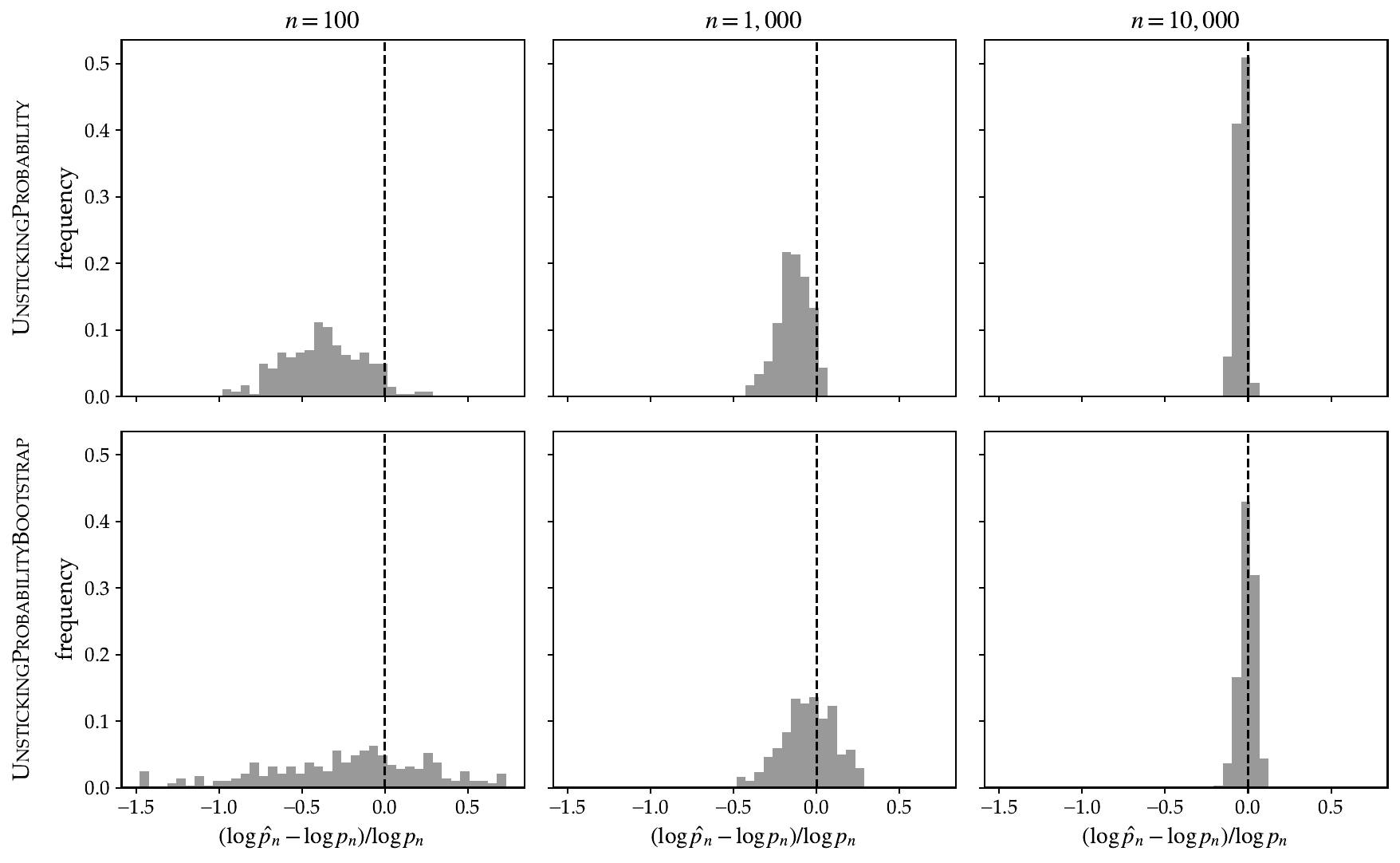}
		\caption{Estimating the unsticking probability of Fr\'echet means using the procedures of Algorithm~\ref{alg:unstick}.
		For $n\in\{10^2,10^3,10^4\}$ samples, we show a histogram of relative errors of $\textsc{UnstickingProbability}$ (top row) and of $\textsc{UnstickingProbabilityBootstrap}$ (bottom row) with respect to the ground truth from Example~\ref{ex:radial}.}
		\label{fig:bootstrap}
	\end{figure}
	
	\section{Applications in Billera-Holmes-Vogtmann Treespace}\label{sec:BHV}
	
	In this section, we extend the methodology of Subsection~\ref{subsec:estim} to the case of Billera-Holmes-Vogtmann (BHV) treespaces, which we then apply to a concrete biological problem involving the phylogeny of tree shrews.
	More precisely, in Subsection~\ref{subsec:BHV-def} we introduce the BHV geometry, in Subsection~\ref{subsec:BHV-methods} we define our methods of interest, and in Subsection~\ref{subsec:appl-phylogen} we apply our methods to the publicly-available data of \cite{SongMammals}.	
	We do not have any theorems in this section; rather, we develop straightforward extensions of the theory of the previous section in order to guide our methodology.
	
	\subsection{Basic Definitions}\label{subsec:BHV-def}
	
	First, we introduce the basic geometric principles of the Billera-Holmes-Vogtmann (BHV) treespace which are needed in order to extend our previous results from the setting of open books.
	We primarily follow the notation from \cite{BardenLe}, but we direct the reader to \cite{BHV} for further detail.
	
	Fix an integer $N\ge3$ and write $[N]=\{1,\dots,N\}$, which we interpret as the set of \emph{leaves}.
	A \emph{clade} is a subset $A\subseteq[N]$ with $2\le|A|\le N-1$, and two clades are called \emph{compatible} if one contains the other or if they are disjoint.
	A \emph{topology} is a set of pairwise compatible clades; a topology is called \emph{binary} or \textit{fully resolved} if every internal node is a \textit{bifurcation}, i.e., has exactly two children, in which case the topology is maximal and has exactly $N-2$ clades.
	Otherwise, the topology is called \textit{unresolved}, meaning that it has at least one internal node with three or more children, called a \textit{multifurcation}; an internal node with exactly three children is called a \textit{trifurcation}.
	
	For each topology $\tau$, we write $\mathcal{O}_\tau:= [0,\infty)^{\tau}$ for the set of all functions from $\tau$ to $[0,\infty)$, which can be interpreted as assigning a nonnegative length to each internal edge of a rooted tree whose topology is $\tau$.
	We define $\TT_N$ to be the set $\bigsqcup_{\textnormal{topology }\tau}(\{\tau\}\times\mathcal{O}_\tau)$
	with the points $(\tau,\ell)$ and $(\tau',\ell')$ identified whenever
	\begin{equation*}
	\{A\in\tau:\ell(A)>0\}=\{A\in\tau':\ell'(A)>0\}\qquad\text{and}\qquad \ell(A)=\ell'(A)\ \text{ for every such } A,
	\end{equation*}
	that is, whenever the two length functions agree on the values assigned to all positive-length clades; note that setting an edge length to zero has the effect of contracting that edge, merging its two endpoints into a single node of higher degree and hence turning a bifurcation into a trifurcation.
	We write $o\in\TT_N$ for the image of $0\in[0,\infty)^{\tau}$ under this identification, which corresponds to the tree with all edge lengths equal to zero.
	We write $y=(\tau,\ell)$ for a tree in $\TT_N$, which parallels our notation $y=(j,v)$ for points in the open book $\B_{K,m}$.
	There is a natural notion of geodesic distance in $\TT_N$ which leads to a metric $d:\TT_N\times\TT_N\to [0,\infty)$ that makes $(\TT_N,d)$ into a complete metric space of nonpositive Alexandrov curvature \cite[Lemma~4.1]{BHV}, although we do not give the precise formulas here.
	Lastly, we write $\S_N\subseteq\TT_N$ for the set of all trees with at least one internal edge of length zero, i.e., the trees with an unresolved topology, containing at least one multifurcation; the set $\S_N$ plays the role of the spine from the case of the open books, although we emphasize that $(\S_N,d)$ is itself a non-Euclidean metric space, unlike the case of the open book $\B_{K,m}$ where the spine is (isometric to) $\Rbb^{m-1}$.
	
	Next, we introduce some finer geometric structure of $\S_N$.
	An unresolved topology $\sigma$ is called \textit{codimension-one} if it has exactly $N-3$ clades, which corresponds to $\sigma$ having exactly one multifurcation and that multifurcation being a trifurcation; in this case there are exactly three binary topologies which are at least as fine as $\sigma$, called \textit{resolutions} of $\sigma$, and the set of all resolutions of $\sigma$ is denoted $R(\sigma)$.
	Each resolution corresponds to adding exactly one internal edge of positive length, called the \textit{resolving edge} of the resolution, which splits the trifurcation into two bifurcations.
	When $\sigma$ is codimension-one, the three orthants $\{\mathcal{O}_\tau\}_{\tau\in R(\sigma)}$ meet pairwise along the suborthant $\mathcal{O}_\sigma\cong[0,\infty)^{N-3}$, so their union $\B_{3,N-2}(\sigma):=\bigcup_{\tau\in R(\sigma)}\mathcal{O}_\tau$ is isometric to a geodesically convex subset of the open book $\B_{3,N-2}$ (namely, the subset in which the spine coordinates are restricted to the orthant $[0,\infty)^{N-3}$); for each resolution $\tau\in R(\sigma)$, we define the folding map $\Phi_{\tau}:\B_{3,N-2}(\sigma)\to \Rbb\times \mathcal{O}_{\sigma}$ to record the signed length of the resolving edge together with the lengths of the edges of $\sigma$.
	In order to lighten the notation, we usually write the codimension-one topology $\sigma$ as a superscript and we abbreviate its resolutions in terms of the leaf sets $A,B$, and $C$ of the three subtrees meeting at the trifurcation, i.e., if an unresolved topology $\sigma$ involves a trifurcation at $\{A,B,C\}$, then $\mathcal{O}^{\sigma}_{\{A,\{B,C\}\}}$ denotes the orthant whose topology has resolving edge separating $A$ from $B\cup C$.
	More generally, one can define $(\Phi_\tau(Y))_1$ for any $Y\in\TT_N$ as the signed rate at which the geodesic from $\mathcal{O}_\sigma$ to $Y$ changes the length of the resolving edge of $\tau$; we do not precisely describe this general folding map here, and instead refer the reader to \cite{BardenOwenLe} for detail.
	See Figure~\ref{fig:BHV-folding} for an illustration of these folding maps.
	
		\begin{figure}
		\begin{tikzpicture}
		
		\def\shift{5}
		\def\ts{0.45}
		
		\filldraw[thin, opacity=0.65, color=black!30] (-\shift, 0) to (-\shift, 3) to (-1.75-\shift,3) to (-1.75-\shift, 0) to cycle;
		\draw[thin, color=black] (-\shift, 0) to (-\shift, 3) to (-1.75-\shift,3) to (-1.75-\shift, 0) to cycle;
		\filldraw[thin, opacity=0.65, color=black!30] (-\shift, 0) to (-\shift, 3) to (1.5-\shift,4) to (1.5-\shift, 1) to cycle;
		\draw[thin, color=black] (-\shift, 0) to (-\shift, 3) to (1.5-\shift,4) to (1.5-\shift, 1) to cycle;
		\treeBeta{(0.3-\shift,1.75)}{\ts}{pythonorange}
		\filldraw[thin, opacity=0.65, color=black!30] (-\shift, 0) to (-\shift, 3) to (1.25-\shift,2) to (1.25-\shift, -1) to cycle;
		\draw[thin, color=black] (-\shift, 0) to (-\shift, 3) to (1.25-\shift,2) to (1.25-\shift, -1) to cycle;
		
		\node at (-0.5-\shift,3.5) {$\TT_{5}$};
		\node at (-2.6-\shift,0.5) {\footnotesize $\mathcal{O}^{\sigma}_{\tiny \{\{a,b\},c\}}$};
		\node at (2.1-\shift,0) {\footnotesize $\mathcal{O}^{\sigma}_{\tiny \{a,\{b,c\}\}}$};
		\node at (2.35-\shift,1.5) {\footnotesize $\mathcal{O}^{\sigma}_{\tiny \{\{a,c\},b\}}$};
		
		\treeGamma{(0.2-\shift,0.3)}{\ts}{pythongreen}
		\treeAlpha{(-1.5-\shift,0.3)}{\ts}{pythonblue}
		
		\treeSigma{(-5.5,-1.35)}{\ts}{black}
		\node[left,font=\small] at (-5.6,-0.945) {$\sigma$};
		
		\filldraw[thin, opacity=0.65, color=black!30] (0, 0) to (0, 3) to (-1,4) to (-1, 1) to cycle;
		\draw[thin, color=black] (0, 0) to (0, 3) to (-1,4) to (-1, 1) to cycle;
		\filldraw[thin, opacity=0.65, color=black!30] (0, 0) to (0, 3) to (1.5,2.5) to (1.5, -0.5) to cycle;
		\draw[thin, color=black] (0, 0) to (0, 3) to (1.5,2.5) to (1.5, -0.5) to cycle;
		\treeBeta{(-0.9,1.75)}{\ts}{pythonorange}
		\filldraw[thin, opacity=0.65, color=black!30] (0, 0) to (0, 3) to (-1.75,3) to (-1.75, 0) to cycle;
		\draw[thin, color=black] (0, 0) to (0, 3) to (-1.75,3) to (-1.75, 0) to cycle;
		\node[color=white] at (0,-0.75) {$\mathcal{B}_{3,3}$};
		
		\treeGamma{(0.2,0.4)}{\ts}{pythongreen}
		\treeAlpha{(-1.5,0.3)}{\ts}{pythonblue}

		\filldraw[thin, opacity=0.65, color=black!30] (\shift, 0) to (\shift, 3) to (-1.75+\shift,3) to (-1.75+\shift, 0) to cycle;
		\draw[thin, color=black] (\shift, 0) to (\shift, 3) to (-1.75+\shift,3) to (-1.75+\shift, 0) to cycle;
		\filldraw[thin, opacity=0.65, color=black!30] (\shift, 0) to (\shift, 3) to (1.75+\shift,3) to (1.75+\shift, 0) to cycle;
		\draw[thin, color=black] (\shift, 0) to (\shift, 3) to (1.75+\shift,3) to (1.75+\shift, 0) to cycle;
		\node[color=white] at (\shift,-0.75) {$\mathcal{B}_{3,3}$};
		\node at (-1.7+\shift,3.35) {\footnotesize$(-\infty,0]\times\mathcal{O}_{\sigma}$};
		\node at (1.6+\shift,3.35) {\footnotesize $[0,\infty)\times\mathcal{O}_{\sigma}$};
		
		\treeAlpha{(-1.5+\shift,0.3)}{\ts}{pythongreen}
		\treeGamma{(0.2+\shift,0.5)}{\ts}{pythongreen}
		\treeBeta{(-1.25+\shift,1.5)}{\ts}{pythongreen}
		
		\draw [thick, -stealth, color=pythongreen] (-3, -0.5) to [out=-30,in=210] (3,-0.5);
		\node[color=pythongreen] at (0, -1) {$\Phi^{\sigma}_{\{a,\{b,c\}\}}$};
	\end{tikzpicture}
	\caption{Visualization of the BHV treespace and the folding map.
		We show (left) the BHV treespace on $N=5$ leaves, locally near the codimension-one topology $\sigma=\{\{a,b,c\},\{d,e\}\}$ which is represented by the black tree whose trifurcation is emphasized by the $\ast$ marker.
		Then, we show (center) the folding map $\Phi_{\{a,\{b,c\}\}}^{\sigma}$ which (right) identifies $\mathcal{O}^{\sigma}_{\{a,\{b,c\}\}}$ with $[0,\infty)\times\mathcal{O}_{\sigma}$ and identifies both $\mathcal{O}^{\sigma}_{\{\{a,b\},c\}}$ and $\mathcal{O}^{\sigma}_{\{\{a,c\},b\}}$ with $(-\infty,0]\times\mathcal{O}_{\sigma}$.}
	\label{fig:BHV-folding}
	\end{figure}
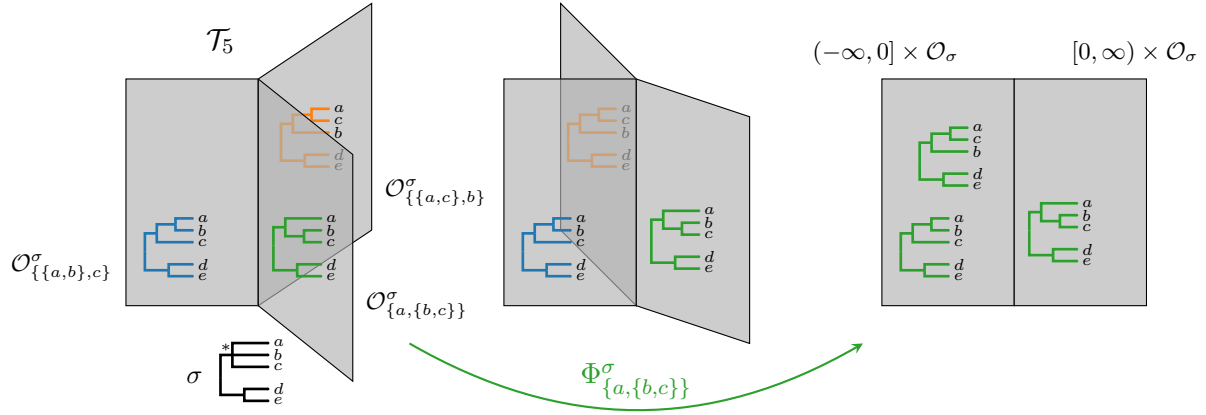

	Fr\'echet means in the metric space $(\TT_N,d)$ are defined similarly as in the case of open books.
	That is, suppose $\mu$ is a probability measure on $\TT_N$ satisfying $\Ebb_{\mu}[d(Y,o)]<\infty$, and let $Y_1,Y_2,\ldots$ denote i.i.d. samples from $\mu$.
	Then,
	\begin{equation*}
		F_{\mu} := \underset{x\in\TT_N}{\arg\min}\,\Ebb_{\mu}\left[d^2(x,Y)-d^2(o,Y)\right]
	\end{equation*}
	exists and is unique, and 
	\begin{equation*}
		F_n := \underset{x\in\TT_N}{\arg\min}\,\frac{1}{n}\sum_{i=1}^{n}d^2(x,Y_i),
	\end{equation*}
	which exists and is unique for all $n\in\Nbb$ almost surely, and we have $\Pbb_{\mu}(F_n\to F_{\mu} \textnormal{ as }n\to\infty)=1$.
	Then, a well-known sufficient condition for stickiness (e.g., \cite[Theorem~3]{BardenLe} or \cite[Section~4]{BardenOwenLe}) is
	\begin{equation}\label{eqn:codim1-stickiness}\tag{$\mathrm S_\sigma$}
		\max_{\tau\in R(\sigma)}\;\Ebb_\mu\left[(\Phi_\tau(Y))_1\right]<0,
	\end{equation}
	i.e., \eqref{eqn:codim1-stickiness} implies $\Pbb_{\mu}(F_n\in \S_N \textnormal{ for sufficiently large }n\in\Nbb)=1$.
	In words, condition~\eqref{eqn:codim1-stickiness} states that the signed length of the resolving edge of $\Phi_{\tau}(Y)$ is negative on average, for all $\tau\in R(\sigma)$.
	As in the case of open books, we write $T:=\inf\{n\in\Nbb: F_m\in\S_N \textnormal{ for all } m\ge n\}$ for the time at which stickiness occurs.
	
	We also revisit the simulation of Figure~\ref{fig:treespace-sim} to explain the sticking time in terms of the collection of folded random walks.
	In this simulation, the trees $Y_1,Y_2,\ldots$ are generated by first sampling a tree topology and then by assigning independent exponential weights to each edge, for a given distribution over topologies and a given set of expected edge lengths.
	The Fr\'echet mean tree $F_n$ has an unresolved topology with a trifurcation at $\{a,b,c\}$, and the three folded random walks correspond to the three possible resolutions of this trifurcation.
	
	\subsection{Description of Methods}\label{subsec:BHV-methods}

In this subsection, we describe extensions of our statistical methodologies from open books to the BHV treespace, which we will use in our phylogenetic application.
For context, we note that there is by now a mature theory for numerical and combinatorial algorithms in BHV treespace \cite{Owen, OwenProvan, MillerOwenProvan, FeragenTreeStats}, and that we will utilize many of these procedures in our methods.

To begin, we give some statistical motivation for estimating the quantity $\Pbb_{\mu}(T>n\,|\,Y_1,\ldots,Y_n)$.
Indeed, suppose that the sample Fr\'echet mean tree $F_n$ has a codimension-one topology $\sigma$, meaning it has a single multifurcation which is a trifurcation; in particular, $F_n\in\mathcal{O}_\sigma\subseteq\S_N$.
Since $T$ is the time after which the sample Fr\'echet mean remains in $\S_N$ forever, the event $\{T>n\}$ is exactly the event that $F_m\notin\S_N$ for some $m\ge n$, i.e., that the observed trifurcation will be resolved at some future time as more gene trees are collected.
As such, the quantity $\Pbb_{\mu}(T>n\,|\,Y_1,\ldots,Y_n)$ is the conditional probability, given the gene trees observed so far, that the trifurcation at $\sigma$ will resolve in the future.
In particular, a value close to $1$ is evidence that the observed trifurcation is merely an artifact of finite sample size (a soft polytomy), while a small value is evidence that it reflects a genuine multifurcation in the population Fr\'echet mean tree (a hard polytomy).
(One could make this interpretation precise by phrasing it as a hypothesis testing problem, with null hypothesis $H_0$ that $F_\mu$ is resolved and alternative $H_\sigma$ that $F_\mu$ has a trifurcation at $\sigma$; we will not pursue this formalization here.)

The estimation procedure is given in Algorithm~\ref{alg:unstick-bhv} which closely parallels Algorithm~\ref{alg:unstick} from the open book setting.
Given gene trees $Y_1,\ldots,Y_n\in\TT_N$ and a codimension-one topology $\sigma$ of interest, we first apply the folding maps $\Phi_\tau$, $\tau\in R(\sigma)$, from Subsection~\ref{subsec:BHV-def} to obtain, for each of the three resolutions $\tau\in R(\sigma)$, a real-valued sequence $(\Phi_\tau(Y_1))_1,\ldots,(\Phi_\tau(Y_n))_1$ recording the signed length of the corresponding resolving edge.
From here, the procedure is identical to \textsc{UnstickingProbability}: we compute the Cram\'er root $\hat\theta_\tau$ associated to each folded sequence, and return $\exp\big(\max_{\tau\in R(\sigma)}\hat\theta_\tau\sum_{i=1}^n(\Phi_\tau(Y_i))_1\big)$ as our estimate of $\Pbb_{\mu}(T>n\,|\,Y_1,\ldots,Y_n)$.
The only substantive difference from the open book case is combinatorial rather than statistical: computing $(\Phi_\tau(Y_i))_1$ for a general gene tree $Y_i\in\TT_N$ requires locating the geodesic from $Y_i$ to $\mathcal{O}_\sigma$ in $\TT_N$, which is a nontrivial computation in BHV treespace (unlike in an open book, where the folding maps are given by an explicit closed-form projection); we use the algorithm of \cite{OwenProvan} for this step.

As before, this estimator suffers from a systematic bias due to the maximum over $\tau\in R(\sigma)$, and we correct for it using the same bootstrap idea as in \textsc{UnstickingProbabilityBootstrap}: we resample the gene trees $Y_1,\ldots,Y_n$ with replacement, recompute the estimate on each resampled dataset, and subtract off the resulting estimate of the bias.
We use this bootstrap-corrected procedure in our phylogenetic application below.
			\begin{algorithm}[h!]
		\caption{Procedures for estimating the unsticking probability for Fr\'echet means in BHV treespace, which are analogous to those of Algorithm~\ref{alg:unstick}.}
		\label{alg:unstick-bhv}
		\begin{algorithmic}[1]
			
			\Procedure{\textsc{UnstickingProbability}}{$Y_1,\ldots,Y_n;\ \sigma$}
			
			\State \textbf{input:} samples $Y_1,\ldots,Y_n\in\TT_N$ and a codimension-one topology $\sigma$, with trifurcation at
			\State \hphantom{\textbf{input:}}\  $\{A,B,C\}$ and resolutions $R(\sigma)=\{\{\{A,B\},C\},\{A,\{B,C\}\},\{\{A,C\},B\}\}$
			\State \textbf{output:} estimate of $\Pbb_{\mu}(T>n\,|\, Y_1,\ldots,Y_n)$
			
			\State \,
		
			\For{$\tau\in R(\sigma)$}
				\State $\hat{\theta}_{\tau}\leftarrow \max\{\lambda\in\Rbb: \frac{1}{n}\sum_{i=1}^{n}\exp(\lambda(\Phi_{\tau}(Y_i))_1)=1\}$
				\EndFor
				
				\State \textbf{return} $\exp\big(\max_{\tau\in R(\sigma)}\hat{\theta}_{\tau}\sum_{i=1}^{n}(\Phi_{\tau}(Y_i))_1\big)$

			\EndProcedure
			
			\State \,
			\State \,
			
			\Procedure{\textsc{UnstickingProbabilityBootstrap}}{$Y_1,\ldots,Y_n;\ \sigma,\ B$}
			
			\State \textbf{input:} samples $Y_1,\ldots,Y_n\in\TT_N$, a codimension-one topology $\sigma$ as
			\State \hphantom{\textbf{input:}}\, above, and a number of replications $B$
			\State \textbf{output:} bias-corrected estimate of $\Pbb_{\mu}(T>n\,|\, Y_1,\ldots,Y_n)$
			
			\State \,
			
			\State $\hat p\leftarrow$ \textsc{UnstickingProbability}$(Y_1,\ldots,Y_n;\ \sigma)$
			\State $\hat g\leftarrow \tfrac{1}{n}\log\hat p$
			
			\For{$b=1,\ldots,B$}
			\State $Y_1^{(b)},\ldots,Y_n^{(b)}\leftarrow$ sample $n$ trees with replacement from $Y_1,\ldots,Y_n$
			\State $g_b\leftarrow \tfrac{1}{n}\log \textsc{UnstickingProbability}(Y_1^{(b)},\ldots,Y_n^{(b)};\ \sigma)$
			\EndFor
			\State $\widehat{\mathrm{bias}}\leftarrow n(\tfrac{1}{B}\sum_{b=1}^{B} g_b - \hat g)$
			
			\State \textbf{return} $\min(\exp(\log \hat p-\widehat{\mathrm{bias}}),\,1)$
			\EndProcedure
		\end{algorithmic}
	\end{algorithm}
	
	\subsection{Phylogeny of Tree Shrews}\label{subsec:appl-phylogen}
	
	Finally, we consider applying our methodology to a concrete biological problem of determining the phylogeny of a collection of eutherian mammals, and in particular determining the phylogeny of the tree shrew relative to other eutherian mammals.
	
	Our data come from \cite{SongMammals}, as processed by \cite{MirarabAstral}, and consists of $n=424$ phylogenetic trees on $N=14$ taxa. The 14 taxa comprise nine primates together with tree shrew, rat, rabbit, horse, and sloth. Each of the 424 trees is a gene tree (i.e., a phylogenetic tree inferred from the sequence alignment of a particular genomic locus across the 14 taxa). Owing to both statistical error and biological processes (see, e.g., \cite{Maddison,DegnanRosenberg}), different loci may yield distinct inferred gene trees, including trees with different topologies. In this setting, Fr\'echet means in BHV treespace provide a method for aggregating the collection of gene trees into a single representative tree, which may be interpreted as an estimate of the underlying species tree, i.e., the evolutionary relationships among the taxa.
		
	Figure~\ref{fig:primate_mean} shows the Fr\'echet mean tree of all of the data.
	Broadly speaking, its topology reflects well-known evolutionary groups (e.g., the great apes $\{\{\{\text{Human},\text{Chimpanzee}\},\text{Gorilla}\},\text{Orangutan}\}$ within a broader clade of primates).
	However, as we discussed in the introduction, there is a trifurcation $\{\text{Tree Shrew},\text{Glires},\text{Primates}\}$ which may be interesting from a biological point of view; there is some debate (see \cite{TreeShrewsPosition,LinShrews}) on the exact taxonomy of tree shrews, which are known to be closely related to both glires and primates.
	
	Applying our methods to this dataset yields the following.
	First, we may visualize this trifurcation using the folded random walks, as in Figure~\ref{fig:primate_RW}; all three folded walks are negative at the endpoint (as required by the trifurcation), but the one corresponding to the resolution $\{\{\textnormal{Tree Shrew},\textnormal{Glires}\},\textnormal{Primates}\}$ is closer to zero than the others, which suggests it has a higher probability of returning to the positive half-line.
	Second, we may apply the procedures of Algorithm~\ref{alg:unstick-bhv} to gain a quantitative understanding of the probability of future bifurcation; more concretely, the \textsc{UnstickingProbabilityBootstrap} procedure from Algorithm~\ref{alg:unstick-bhv} returns a value of 0.02795, which (recall the logarithmic scaling in Proposition~\ref{prop:alg-consistency}) implies that the probability of future bifurcation is on the order of $10^{-2}$.

	\begin{figure}[t]
		\centering
		\includegraphics[scale=0.75]{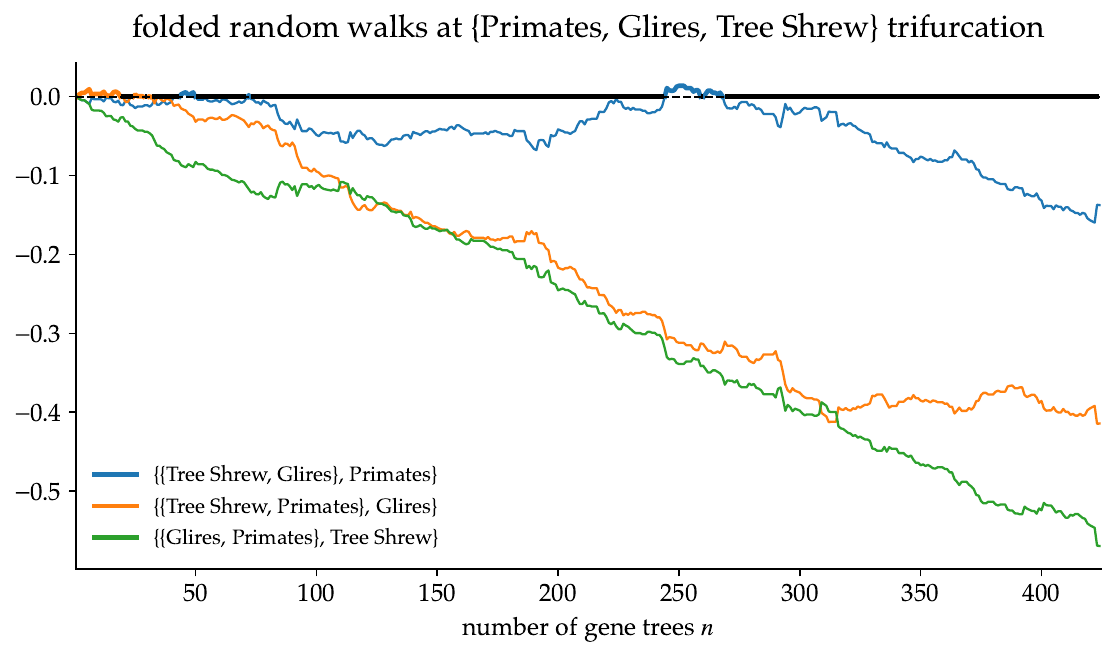}
		\caption{The three folded random walks for the gene trees of eutherian mammals, at the \{Primates, Glires, Tree Shrew\} trifurcation.
		All three random walks are negative at the endpoint, which reflects the trifurcation in Figure~\ref{fig:primate_mean}.
		The random walk closest to zero is the one in which tree shrews are more closely related to glires than primates.}
		\label{fig:primate_RW}
	\end{figure}

	\subsection*{Acknowledgements}
	
	Claude Opus 4.8 was used for various tasks in the preparation of this manuscript, e.g., writing \texttt{tikz} code, proofreading language and mathematical details, and locating open-source data sets.
	Additionally, Claude Fable 5 was used to translate Megan Owen's software packages for BHV treespace algorithms (available at \url{http://comet.lehman.cuny.edu/owen/code.html}) from \texttt{Java} to \texttt{Python}.
	All ideas and arguments are due to the author, and the author takes responsibility for the correctness of all results and methods.

	\bibliography{refs}
	\bibliographystyle{alpha}
	
\end{document}